\documentclass[%
 reprint,
 amsmath,amssymb,
 aps,
 prd,
]{revtex4-2}
\usepackage{import}
\usepackage{dcolumn}
\usepackage{amsmath}  
\usepackage{amsfonts} 
\usepackage{textcomp}
\usepackage[utf8]{inputenc} 
\usepackage[T1]{fontenc}    
\usepackage{hyperref}       
\usepackage{url}            
\usepackage{booktabs}       
\usepackage{nicefrac}       
\usepackage{microtype}      
\usepackage{lipsum}
\usepackage[pdftex]{graphicx}
\usepackage{caption}
\usepackage{amssymb, mathtools, amsthm}
\usepackage{mathrsfs}
\usepackage{tensor}
\usepackage[dvipsnames]{xcolor}
\usepackage{todonotes}
\usepackage{float} 
\usepackage{ragged2e}
\usepackage{scalerel}
\usepackage[normalem]{ulem}
\usepackage[english]{babel}
\usepackage{graphicx}
\usepackage{dcolumn}
\usepackage{bm}
\usepackage{blindtext}
\usepackage{verbatim}
\usepackage{relsize}
\usepackage{musicography}
\usepackage{blindtext}
\usepackage{cancel}
\usepackage{physics}
\usepackage{epstopdf}
\usepackage{mathtools}
\usepackage{blindtext}
\usepackage{color}
\usepackage{epsfig}
\usepackage{pst-grad} 
\usepackage{pst-plot} 
\usepackage{hyperref}
\usepackage{verbatim}
\usepackage{slashed}
\usepackage{subcaption}
\usepackage{dsfont}
\usepackage[english]{babel}
\usepackage{cleveref}

\newcommand{\p}{\partial}

\newcommand{\ZZ}{\mathbb{Z}}
\newcommand{\CC}{\mathbb{C}}
\newcommand{\RR}{\mathbb{{R}}}
\newcommand{\NN}{\mathbb{N}}
\newcommand{\FF}{\mathcal{F}}
\newcommand{\WW}{\mathcal{W}}
\newcommand{\GG}{\mathcal{G}}

\newcommand{\Ebar}{\overline{E}}
\newcommand{\Eabs}{|\overline{E}|}

\newcommand{\qmn}{q_{|m|n}}

\newcommand{\amn}{\alpha_{|m|n}}
\newcommand{\amN}{\alpha_{|m|N}}

\newcommand{\amo}{\alpha_{|m|0}}

\newcommand{\amx}{\alpha_{|m|x}}

\newcommand{\Nmin}{N_{\text{min}}}

\newcommand{\Nmax}{N_{\text{max}}}

\newcommand{\ii}{\mathrm{i}}

\newcommand{\ee}{\mathrm{e}}
\newcommand{\xx}{\mathsf{x}}

\usepackage{float}
\usepackage{tikz-cd}
\usepackage{amsthm}
\usepackage{geometry}
\usepackage{braket}
\DeclareMathOperator{\Ai}{Ai}

\DeclareMathOperator{\supp}{supp}

\hypersetup{  colorlinks,
              linktoc         = page, 
              urlcolor        = {green!50!blue},
              linkcolor       = {Bittersweet!80!black}, 
              urlcolor        = {Bittersweet!80!black}, 
              citecolor       = {cyan!50!black}, 
              anchorcolor     = {yellow}
}

\numberwithin{equation}{section}

\theoremstyle{plain}
\newtheorem{theo}{Theorem}[section]

\newtheorem{propo}[theo]{Proposition}

\theoremstyle{definition}

\newcommand{\mrm}[1]{\mathrm{#1}}
\newcommand{\mc}[1]{\mathcal{#1}}
\newcommand{\wt}[1]{\widetilde{#1}}

\newcommand{\lr}[1]{\!\left(#1\right)}

\newcommand{\lrb}[1]{\!\left[#1\right]}

\newcommand{\lrl}[2]{\!\left.#1\right|_{#2}}

\usepackage[final]{showlabels}

\newcommand{\td}[2]{\ensuremath{\frac{\dd #1}{\dd #2}}}
\newcommand{\tdn}[3]{\ensuremath{\frac{\dd^{#3} #1}{\dd #2^{#3}}}}

\newcommand{\beq}{\begin{equation}}
\newcommand{\eeq}{\end{equation}}

\newcommand{\ba}{\begin{align}}
\newcommand{\ea}{\end{align}}

\allowdisplaybreaks[1] 

\begin{document}

\title{Circular acceleration in Minkowski spacetime: thermality versus finite size}

\author{Cameron R D Bunney}
\email{cameron.r.d.bunney@gmail.com}
\affiliation{School of Mathematical Sciences and School of Physics and Astronomy, 
University of Nottingham, University Park, Nottingham, NG7 2RD, UK}

\author{Jorma Louko}
\email{jorma.louko@nottingham.ac.uk}
\affiliation{School of Mathematical Sciences, University of Nottingham, University Park, Nottingham, NG7 2RD, UK}

\date{September 2026}

\begin{abstract}
The Unruh effect predicts that a uniformly linearly accelerated observer with proper acceleration $a$ reacts to the Minkowski vacuum through excitations and de-excitations with the characteristics of a thermal state at temperature $T_U=a/(2\pi)$. An observer in uniform circular motion will experience similar excitations and de-excitations that we may use to operationally define an effective temperature, which however depends not only on the acceleration but also on the orbital speed and excitation energy. Motivated by the experimental interest in the circular motion Unruh effect, we investigate how spatial confinement modifies the response of an Unruh-DeWitt detector in $2+1$ Minkowski spacetime. We consider a massless scalar field confined within a circular boundary prepared in either the vacuum or a thermal state, probed by an Unruh-DeWitt detector on a circular orbit, a setting that describes proposed analogue spacetime systems for testing the effect, and in which both a boundary and an ambient temperature will necessarily be present. We establish analytic results for the detector response in the large-boundary regime and identify resonance peaks, which are more prominent when the field has an ambient temperature.
\end{abstract}

\maketitle

\section{Introduction}
The Unruh effect~\cite{Fulling,Davies1975,Unruh} is the prediction in quantum field theory that a uniformly linearly accelerated observer with proper acceleration $a$ reacts to the Minkowski vacuum as if in a thermal state at the Unruh temperature,
\begin{equation}\label{eq:T Unruh}
    T_U~=~\frac{\hbar a}{2\pi c k_{\mrm{B}}}\,.
\end{equation}
The effect is still lacking a direct experimental confirmation, in large part due to the sheer magnitude of acceleration required to reach a detectable increase in temperature. However, experimental confirmation of the Unruh effect continues to attract broad interest due to its connection to predictions such as the Hawking effect~\cite{Hawking:1975vcx}, and the quantum effects in the early Universe, from which the present-day large structure of the Universe may originate~\cite{Parker:1969au,Mukhanov:2007zz}. 

Over the last two decades, gravity simulators have opened up new avenues for exploring curved-spacetime phenomena in analogue spacetimes provided by hydrodynamical and condensed-matter systems~\cite{Unruh1981,Analogue,HeliumUniverse,SchutzholdUnruh,Vitor_EFT}. In this setting, analogue horizons have been observed both in an optical fibre~\cite{opticalHorizon} and in open-channel water flow~\cite{RousseauxHorizon}, and the stimulated Hawking spectrum has been observed~\cite{WeinfurtnerHawking}. More recently, several black-hole phenomena have been observed in $(2+1)$-dimensional hydrodynamical systems: superradiance around an analogue rotating black hole~\cite{TorresSuperradiance}; the relaxation of an analogue black hole through the emission of quasinormal modes~\cite{Ringdown}; and the observation of bound states in a giant superfluid helium vortex flow~\cite{Svancara_Helium_Vortex}. In these nonrelativistic laboratory systems, excitations (typically in the density, velocity potential, or fluid height) simulate a relativistic quantum field, with the speed of light replaced by the speed of sound in the system. The Unruh temperature~\eqref{eq:T Unruh} in an analogue spacetime is, therefore, higher by several orders of magnitude.

Unruh-like phenomena also exist for non-linear uniformly accelerated motions~\cite{Letaw,Korsbakken:2004bv,Good:2020hav}. 
In this paper we consider uniform circular motion, 
whose unique experimental selling point is that the accelerated trajectory may be kept within a finite experimental volume for an arbitrarily long time. 
An early proposal to observe the circular motion Unruh effect was to consider spin depolarisation in accelerator storage rings \cite{BellLeinaas,Bell:1986ir,Unruh:1998gq,Leinaas:1998tu}. 
A more recent proposal is to use tabletop condensed matter analogue spacetime systems, such as Bose-Einstein condensates or thin-film superfluid helium \cite{BEC1,Marino:2020uqj,Gooding,BunneySounds,Gooding:2025tfp}. While these tabletop analogue spacetime systems do not naturally incorporate relativistic time dilation, an advantage of uniform circular motion in these systems is that the time-dilation factor between the trajectory and the laboratory is a time-independent constant, and this constant may be accounted for during the data analysis~\cite{Biermann}. 
The theory of the circular motion Unruh effect has been analysed in a range of relativistic spacetime and analogue spacetime contexts  in~\cite{LetawPfautsch,Letaw,Takagi,RotatingVacuum,Korsbakken:2004bv,BibhasRickpaper,Good:2020hav,BunneyThermal,BunneyRick,Bunneydsads,ParryDRCM}.

In this paper, we address the circular motion Unruh effect for a quantum field that is confined to a finite spatial volume, in a setting that is relevant for the proposed tabletop experiments of \cite{BEC1,Marino:2020uqj,Gooding,BunneySounds,Gooding:2025tfp}. For context, we recall that an early and well-known example of finite-size effects in quantum field theory is the Casimir effect \cite{CasimirOriginal,CasimirPolderOriginal,birrell,CasimirNotVac}, which was experimentally confirmed to high accuracy in 1997~\cite{CasimirExperiment}. In curved spacetime, related phenomena occur in asymptotically anti-de Sitter (AdS) spacetimes, where a negative cosmological constant introduces a notion of confinement, and boundary conditions at the infinity are necessary to ensure unitary quantum dynamics~\cite{AdSQFT}; well-studied examples include global AdS spacetime \cite{AdSQFT} and the Ba\~nados-Teitelboim-Zanelli black hole~\cite{BTZoriginal,BTZgeometry,CarlipBTZ}. 

Spatial confinement also influences related dynamical processes in field theory, 
such as the emission of quasinormal modes by black holes, which has been studied both in relativistic spacetimes with a negative cosmological constant \cite{BTZQNM} and in analogue spacetime models with spatial confinement~\cite{SolidoroQNM}. 
A~superfluid helium vortex has recently provided an experimental setting for investigating these effects in curved analogue spacetime~\cite{SmaniottoQNMExp}.

The specific system in which we investigate the circular motion Unruh effect in this paper is $(2+1)$-dimensional flat spacetime in a cavity with a static circular boundary. We take the ambient quantum field to be initially prepared either in the vacuum state or in a thermal state. Related analyses in $3+1$ dimensions with an initial vacuum state are given in~\cite{RotatingVacuum,Marino:2020uqj,KinjalkFiniteSize}.

We take the uniform circular motion to be about the centre of the cavity. This system is invariant under time translations along the circular trajectory, both for an initial vacuum state and for an initial thermal state, and the Unruh effect will be time independent. We take the quantum field to be a real massless scalar field, and we impose at the boundary the Dirichlet boundary condition. We probe the field with a pointlike Unruh-DeWitt (UDW) detector~\cite{Unruh,DeWitt:1980hx}. In contrast to the usual relativistic setting, in which the detector's energy levels are defined with respect to the detector's relativistic proper time~\cite{Unruh,DeWitt:1980hx}, we define the detector's energy levels with respect to the Minkowski time in the cavity's rest frame. This is the setting relevant for analogue spacetime systems, in which the signal of the Unruh effect will be measured in terms of the laboratory time~\cite{Gooding}. 

We wish to quantify the effects due to the finite size of the cavity, by considering how rapidly the detector's response in the cavity tends to the response in full $2+1$ Minkowski space when the cavity radius approaches infinity, both when the field's initial state is the vacuum and when the field's initial state is a thermal state. 
To this end, we consider a detector that is coupled linearly to the Minkowski-time derivative of the scalar field along the detector's trajectory, rather than to the scalar field itself. The reason is that in full $2+1$ Minkowski space, coupling the detector to the time derivative sidesteps a technical issue that arises from the infrared divergence of the thermal state Wightman function~\cite{BunneyThermal}: the derivative-coupled detector yields for us a well-defined point of comparison in the large cavity limit, both in vacuum and in a thermal state, even if the non-derivative coupling would in itself be a well-defined system at any finite size of the cavity. Coupling the detector to the field's time derivative has been employed as a technique to sidestep a similar infrared divergence that occurs for a massless scalar field in $1+1$ dimensions already at zero temperature~\cite{1991RSPSA.435..205R,Raval:1995mb,Wang:2013lex,Juarez-Aubry:2014jba,Juarez-Aubry:2018ofz,Juarez-Aubry:2021tae}. 

To begin, we consider a scalar field that need not be massless and may have an arbitrary dispersion relation, subject to mild monotonicity conditions. 
In an analogue spacetime setting, this allows for field frequencies that go beyond the linear-dispersion regime~\cite{Analogue}; in a relativistic spacetime setting, this allows for dispersion relations that may arise from Planck-scale physics~\cite{Amelino-Camelia:2008aez,LoukoViqar,Upton}. 
We assume that the interaction lasts for a finite amount of time, and we express the detector's response in the form of a mode sum. 

We then specialise to a massless Klein-Gordon field in the limit of a long interaction time. In the large cavity limit, we recover the response in full Minkowski space~\cite{BunneyThermal}, as was to be expected. In the subleading order, the response contains resonance peaks at frequencies determined by the detector's angular velocity, and we find that these peaks are more prominent when the field's initial state is thermal. 

As a mathematical side outcome, we find an asymptotic expansion for the Bessel-function normalisation factor \eqref{eq: new asymptotic} in terms of the McMahon expansion of the large zeros of Bessel functions~\cite{NIST}. 
We have not encountered this expansion in the existing literature.

The paper is structured as follows. We begin in Section \ref{sec2} by establishing the preliminaries for a UDW detector in uniform circular motion within a static circular cavity in $(2+1)$-dimensional flat spacetime, 
coupled to the Minkowski-time derivative of a quantised, real, massless scalar field obeying the Dirichlet boundary condition at the cavity's boundary, with the field initially prepared either in the vacuum state or in a thermal state. 
Section \ref{sec:large boundary} takes the long interaction time limit and establishes the main results about the asymptotics of the detector's response in the limit of large cavity size. Section~\ref{sec:conc} provides a summary and concluding remarks. Technical asymptotic expansions are deferred to two appendices.

We use units in which $c=\hbar=k_{\mrm{B}}=1$. Sans serif letters ($\xx$) denote spacetime points and boldface Italic letters ($\bm{k}$) denote spatial vectors. We use metric signature $(-,+,+)$. In asymptotic formulæ, $f(x) = O(g(x))$ denotes that $f(x)/g(x)$ remains bounded in the limit considered, $f(x)=o(g(x))$ denotes that $f(x)/g(x)$ tends to zero in the limit considered, and $f(x)\sim g(x)$ denotes that $f(x)/g(x)$ tends to unity in the limit considered. We denote the complex conjugate of $f$ by $f^*$ and the Hermitian conjugate of $A$ by $A^\dagger$.

\section{Spacetime, field, and detector preliminaries}\label{sec2}
In this Section, we review the derivative-coupled interaction between an Unruh-DeWitt (UDW) detector and a real scalar field in ($2+1$)-dimensional Minkowski spacetime with a circular boundary. We consider both the case when the field is prepared in the vacuum state and the case when the field is prepared in a thermal state. We work in the limit of weak coupling and long interaction time, with negligible back-action of the detector on the field.

\subsection{Field and detector}\label{sec: field detector}
We work in Minkowski spacetime in $2+1$ dimensions with a standard set of Minkowski coordinates $(t,x,y)$ and the metric
\begin{equation}
    \dd s^2=-\dd t^2+\dd x^2+\dd y^2\,.
\end{equation} We consider a quantised, real scalar field $\Phi$ with a dispersion relation that is isotropic in $(x,y)$ and subject to very mild positivity and asymptotic conditions, which we will specify in~\ref{sec: mode sum}. We do not assume the dispersion relation to be Lorentz invariant. We denote by $\mc{H}_\Phi$ the Fock space in which positive frequencies are defined with respect to the timelike Killing vector $\p_t$.

In polar coordinates $(x,y)=(r\cos\theta,r\sin\theta)$, we model a cylindrical cavity by imposing the Dirichlet boundary condition at some fixed radius $r=a>0$, $\Phi(\xx)|_{r=a}=0$~\cite{RotatingVacuum}. 
We expand the field operator $\Phi$ as
\begin{equation}\label{eq:field modes}
    \Phi(\xx)~=~\sum_{m\in\ZZ}\sum_{n\in\NN}\big(\phi_{mn}(\xx)a_{mn}+\phi^*_{mn}(\xx)a^\dagger_{mn}\big)\,,
\end{equation}
where the Minkowski positive-frequency mode functions subject to the Dirichlet boundary condition are 
\begin{equation}\label{eq:field modes b}
    \phi_{mn}(\xx)~=~\frac{J_{|m|}\lr{\frac{r}{a}\qmn}}{\sqrt{2\pi a^2\omega_{mn}}|J_{|m|+1}(\qmn)|}\ee^{-\ii \omega_{mn}t+\ii m\theta}\,,
\end{equation}
$\NN=\{1,2,3,\dots\}$, 
$J_{|m|}$ is the Bessel function of the first kind of order $|m|$, and $\qmn$ is the $n$th positive zero of $J_{|m|}$~\cite{NIST}. We have written $\omega_{mn}\coloneq \omega(\qmn/a)$, where the function $\omega(|\bm{k}|)$ is the dispersion relation, giving the frequency as a function of the magnitude of the spatial momentum. We assume $\omega(|\bm{k}|)$ to be positive when $|\bm{k}|>0$, 
and we assume it to increase sufficiently fast as $|\bm{k}|\to\infty$ for the mode sums that emerge to be convergent. 
We note that since Bessel functions of different orders do not share zeros~\cite{NIST}, the denominator in \eqref{eq:field modes b} is nonzero and $\phi_{mn}$ are well defined.

The field modes $\phi_{mn}$ are normalised with respect to the Klein-Gordon inner product, $(\phi_{mn},\phi_{m'n'})=\delta_{mm'}\delta_{nn'}$. 
It follows that the annihilation $a_{mn}$ and creation $a^\dagger_{mn}$ operators obey the commutation relation $[a_{mn},a^\dagger_{m'n'}]=\delta_{mm'}\delta_{nn'}$. The Hilbert space is a Fock space with Fock vacuum~$\ket{0}$, satisfying $a_{mn}\ket{0}=0$.

We probe the field using the UDW detector model~\cite{Unruh,DeWitt:1980hx}. The detector's Hilbert space $\mc{H}_{\mrm{D}}\simeq\CC^2$ is spanned by the orthonormal basis $\{\ket{0_{\mrm{D}}},\ket{1_{\mrm{D}}}\}$. The dynamics of the detector are described by the Hamiltonian $H_{\mrm{D}}$ with respect to lab time~$t$, whose action on $\mc{H}_{\mrm{D}}$ is $H_\mrm{D}\ket{0_\mrm{D}}=0$ and $H_\mrm{D}\ket{1_{\mrm{D}}}=\Ebar \ket{1_\mrm{D}}$. The constant $\Ebar\in\RR\setminus\{0\}$ is the energy gap of the detector. For $\Ebar>0$, $\ket{0_\mrm{D}}$ is the ground state and $\ket{1_\mrm{D}}$ is the excited state. For $\Ebar <0$, the roles are reversed. The total Hilbert space of the system is $\mc{H}_\mrm{D}\otimes\mc{H}_\Phi$. We impose $\Ebar\neq0$ as the response function for a detector in inertial motion and circular motion in $2+1$ Minkowski spacetime is discontinuous at $\Ebar=0$ in the long time limit \cite{Biermann,ParryDRCM}. 

We have included in the symbol $\Ebar$ an overline to emphasise that this energy gap is defined with respect to Minkowski time~$t$, rather than relativistic proper time. 
This formulation will give us a detector response that is appropriate for describing an analogue spacetime, where $t$ is the `lab time', with respect to which any frequencies in an experiment will be measured~\cite{Biermann,Gooding,BunneySounds,BunneyThermal,Bunneythesis}. For a detector modelling a localised relativistic quantum system, the energy gap would be defined with respect to the relativistic proper time. 

We couple this two-level quantum system to the Minkowski-time derivative of the field via the monopole moment operator $\mu(t)$ and assume the detector to move along a given timelike worldline $\xx(t)$ parametrised by the Minkowski time~$t$. The interaction Hamiltonian is 
\begin{equation}\label{eq:H int}
    H_{\mrm{int}}~=~\lambda\chi(t)\mu(t)\td{}{t}\Phi(\xx(t))\,,
\end{equation}where $\lambda$ is a coupling constant, $\Phi(\xx(t))$ is the value of the field pulled back to the worldline of the detector, and $\chi$ is a real-valued switching function that specifies how the interaction is turned on and off. We highlight that the switching function $\chi$ only smears the detector in time and not in space. The UDW detector model therefore describes the interaction between the field and a pointlike detector. Spatially smeared, or finite-size, detectors have been considered in, e.g.,~\cite{SmearedDetectors,SmearedDetectors2}.

We assume that the field has been prepared in a thermal state at inverse temperature $\beta>0$, denoted by~$\ket{\beta}$, where the notion of thermality is with respect to the time evolution generated by~$\p_t$. Working to first order in perturbation theory in~$\lambda$, the probability for the detector to transition from $\ket{0_{\mrm{D}}}$ to~$\ket{1_{\mrm{D}}}$, regardless of the final state of the field, is then \cite{Unruh,birrell}
\begin{equation}
    \mc{P}(E)~=~\lambda^2|\braket{1_{\mrm{D}}|\mu(0)|0_{\mrm{D}}}|^2\FF_\chi(\Ebar,\beta)\,,
\end{equation}where $\FF_\chi$ is the response function, given by 
\begin{subequations}
\label{eq:Fchi+Wbeta}
    \begin{align}\label{eq:F chi}
       \FF_\chi(\Ebar,\beta)&=\int_{\RR^2}\!\!\!\dd t'\dd t''\chi(t')\chi(t'')\ee^{-\ii \Ebar(t'-t'')}\WW_\beta(t',t'')\,, \\
       \WW_\beta(t',t'')&=\bigg\langle\td{}{t'}\Phi(\xx(t'))\td{}{t''}\Phi(\xx(t''))\bigg\rangle_\beta\,.\label{eq:W beta}
    \end{align}
\end{subequations}
Here, $\WW_\beta$ is the correlation function of the time derivative of the pullback of the field on the detector's worldline, in the thermal state~$\ket{\beta}$. 
We refer to $\WW_\beta$ \eqref{eq:W beta} as the derivative correlation function. 

The reason we have included the time derivative in the interaction Hamiltonian \eqref{eq:H int} is that we wish to consider the $a\to\infty$ limit. For finite values of~$a$, 
we could consider a detector without the derivative in~\eqref{eq:H int}, and \eqref{eq:W beta} would then be replaced by the thermal Wightman function, without derivatives. 
When the scalar field is massless, however, the thermal Wightman function without derivatives develops an infrared divergence as $a\to\infty$, and this raises technical issues for recovering a meaningful $a\to\infty$ limit in the detector's response~\cite{BunneyThermal}. 
With the time derivative in the interaction Hamiltonian~\eqref{eq:H int}, by contrast, the relevant correlation function is $\WW_\beta$~\eqref{eq:W beta}, whose $a\to\infty$ limit is well defined. 

It follows as in \cite{Takagi} that the derivative correlation function \eqref{eq:W beta} has the mode-sum expression
\begin{multline}\label{eq:W beta mode sum general}
    \WW_\beta(t',t'')~=~\sum_{m,n}\bigg[\td{}{t'}\phi_{mn}(\xx(t'))\td{}{t''}\phi^*_{mn}(\xx(t''))\\
    +n(\beta\omega_{mn})\td{}{t'}\phi_{mn}(\xx(t'))\td{}{t''}\phi^*_{mn}(\xx(t''))\hphantom{\quad\quad\;\:\:}\\
    +n(\beta\omega_{mn})\td{}{t'}\phi^*_{mn}(\xx(t'))\td{}{t''}\phi_{mn}(\xx(t''))\bigg]\,,
\end{multline}where $n(x)$ is the Planckian factor
\begin{equation}\label{eq:n Planck}
    n(x)~=~\frac{1}{\ee^x-1}\,.
\end{equation}The mode-sum expression~\eqref{eq:W beta mode sum general} splits naturally into a contribution from the Fock vacuum $\ket{0}$ plus thermal corrections,
\begin{subequations}\label{eq:W vac therm split}
    \begin{align}\WW_\beta(t',t'')~=&~\WW_\infty(t',t'')+\Delta\WW_\beta(t',t'')\,,\\
        \WW_\infty(t',t'')~=&~\sum_{m,n}\td{}{t'}\phi_{mn}(\xx(t'))\td{}{t''}\phi^*_{mn}(\xx(t''))\,,
        \label{eq:W vac therm split vac}\\
        \Delta\WW_\beta(t',t'')~=&~ \sum_{m,n}n(\beta\omega_{mn}) \nonumber\\
        &\times \bigg[\td{}{t'}\phi_{mn}(\xx(t'))\td{}{t''}\phi^*_{mn}(\xx(t''))\nonumber\\
        &\hphantom{abc\,}+\td{}{t'}\phi^*_{mn}(\xx(t'))\td{}{t''}\phi_{mn}(\xx(t''))\bigg]\,,
        \label{eq:W vac therm split therm}
    \end{align}
\end{subequations}
where $\WW_\infty$ is the contribution from the Fock vacuum, corresponding to the $\beta\to\infty$ limit. 

We specialise now to a detector undergoing uniform circular motion, parametrised in terms of Minkowski time $t$ as  
\begin{equation}
    \xx(t)~=~(t,R\cos(\Omega t),R\sin(\Omega t))\,,
\end{equation}where $R>0$ is the radius and $\Omega=\td{\theta}{t}>0$ is the angular velocity. The orbital speed is given by $v=R\Omega$. We assume that the worldline is timelike, $v<1$. The worldline is an orbit of the helical Killing vector $\partial_t + \Omega^{-1} \partial_\theta$, and is thus stationary~\cite{Letaw,LetawPfautsch,TownsendStationary,BunneyStationary}. As the thermal state $\ket{\beta}$ is invariant under the time translation Killing vector $\partial_t$ and the rotational Killing vector~$\partial_\theta$, the derivative correlation 
function $\mc{W}_\beta(t',t'')$ is invariant under time translations along the detector's worldline~\cite{Biermann}, so that 
\begin{equation}
    \WW_\beta(t',t'')~=~\WW_\beta(t'-t'',0)\,.
\label{eq:Wbeta-stationary}
\end{equation} 

In the response function $\FF_\chi$~\eqref{eq:F chi}, we may divide by the total duration of the interaction and let the duration tend to infinity in a controlled fashion \cite{Fewster,BunneyRick}, 
which we review below in Section~\ref{sec: finite time}. 
In this limit, $\FF_\chi$ reduces to the transition probability per unit Minkowski time, or the transition rate, given by the stationary response function, 
\begin{equation}\label{eq:F def}
    \FF(\Ebar,\beta)~=~\int_{\RR}\dd t\,\ee^{-\ii \Ebar t}\WW_\beta(t,0)\,.
\end{equation}Using the vacuum-thermal decomposition~\eqref{eq:W vac therm split}, 
$\FF(\Ebar,\beta)$ decomposes into a vacuum term and a thermal correction, as 
\begin{subequations}\label{eq:F vac therm split}
    \begin{align}
        \FF(\Ebar,\beta)&~=~\FF_\infty(\Ebar)+\Delta\FF_\beta(\Ebar)\,,\\
        \FF_\infty(\Ebar)&~=~\int_\RR\dd t\,\ee^{-\ii \Ebar t}\WW_\infty(t,0)\,,
        \label{eq:F vac therm split vac}\\
        \Delta\FF_\beta(\Ebar)&~=~\int_\RR\dd t\,\ee^{-\ii \Ebar t}\Delta\WW_\beta(t,0)\,.
        \label{eq:F vac therm split therm}
    \end{align}
\end{subequations}From now on, we will work with the stationary response function $\FF(\Ebar,\beta)$ \eqref{eq:F vac therm split} and refer to it as simply the response function.

Finally, recall that $\Ebar$ was defined as the gap with respect to Minkowski time, and the stationary response function $\FF(\Ebar,\beta)$ \eqref{eq:F vac therm split} is the transition probability per unit Minkowski time, which is appropriate for describing analogue spacetime experiments. 
For a detector modelling a localised relativistic quantum system, the stationary response function that describes transition probability per unit proper time is obtained from $\FF(\Ebar,\beta)$ \eqref{eq:F vac therm split} by writing 
$\Ebar=E/\gamma$, where $E$ is the gap with respect to the detector's proper time, and including the overall factor~$1/\gamma$. 
Note that this conversion is time-independent because the speed of the circular motion is constant in the Minkowski frame $(t,x,y)$.

\subsection{Response function as a mode sum}\label{sec: mode sum}

We now write the response function as a mode sum. 

Using in \eqref{eq:F vac therm split} the field-mode expressions \eqref{eq:field modes} and the vacuum and thermal parts of the derivative correlation function~\eqref{eq:W beta mode sum general}, we find 
\begin{subequations}\label{eq:F vac plus therm}
\begin{align}
\FF(\Ebar,\beta,a)&=\FF_\infty(\Ebar,a)+\Delta\FF_\beta(\Ebar,a)\,,\\
\FF_\infty(\Ebar,a)&=\frac{\Ebar^2}{a^2}\sum_{m,n} \frac{J^2_{|m|}(\frac{R}{a}\qmn)}{\omega_{mn}J^2_{|m|+1}(\qmn)}\nonumber\\
    &\hphantom{abcdefg}\times\delta(\omega_{mn}-m\Omega+\Ebar)\,,
\label{eq:F vac}
\\
    \!\!\!\Delta\FF_\beta(\Ebar,a)&=\frac{\Ebar^2}{a^2}\sum_{m,n}n(\beta\omega_{mn})\frac{J^2_{|m|}(\frac{R}{a}\qmn)}{\omega_{mn}J^2_{|m|+1}(\qmn)}\nonumber\\
    \times\big(\delta(\omega_{mn}&-m\Omega+\Eabs)+\delta(\omega_{mn}-m\Omega-\Eabs)\big)\,,\label{eq:F therm}
\end{align}  
\end{subequations}
where we have included $a$ as the last argument of $\FF(\Ebar,\beta,a)$, $\FF_\infty(\Ebar,a)$ and $\Delta\FF_\beta(\Ebar,a)$ as a reminder of the dependence on the boundary radius~$a$.
We note that $\Delta\FF_\beta(\Ebar,a)$ is even in~$\Ebar$, and we have written \eqref{eq:F therm} in a way that makes this manifest.

The prefactor $\Ebar^2$ in 
\eqref{eq:F vac plus therm}
has come from the time derivatives in \eqref{eq:W beta mode sum general}. Dropping this prefactor would give the response function of a detector whose coupling does not have a time derivative. 
As mentioned in Section~\ref{sec: field detector}, our coupling includes the time derivative because this will enable us to consider the $a\to\infty$ limit in 
Section~\ref{sec:large boundary}. 

Note that \eqref{eq:F vac plus therm} does not assume the dispersion relation to be Lorentz invariant. The expressions in \eqref{eq:F vac plus therm} hence apply to dispersion relations that occur in condensed matter systems, such as the Bogoliubov dispersion relation in a Bose-Einstein condensate and the dispersion relation for height fluctuations in thin-film superfluid helium \cite{Gooding,BunneySounds}. 

In the special case of a massless scalar field, with the Lorentz-invariant dispersion relation 
$\omega(|\bm{k}|) = |\bm{k}|$, there is a simple necessary and sufficient condition for the detector to register excitations in the Fock vacuum: this condition is $a\Omega>1$. 
Recalling that excitations correspond to $\Ebar>0$, and using $\omega_{mn}=\qmn/a$, the argument of the Dirac delta in \eqref{eq:F vac} shows that excitations occur if and only if 
\begin{equation}
0< m\Omega-\frac{\qmn}{a}\,,
\label{eq:masslessvac-excitations}
\end{equation}
for some $m,n$. 
For $m\le0$, \eqref{eq:masslessvac-excitations} does not hold. For $m>0$, \eqref{eq:masslessvac-excitations} can be rearranged as 
\begin{equation}
\frac{q_{m n}}{m} < a\Omega\,. 
\label{eq:masslessvac-excitations-mpos}
\end{equation}
As $q_{m n}$ is increasing in $n$, and as $q_{m n}/m$ is decreasing in $m$ and satisfies $q_{m n}/m \to 1$ as $m\to\infty$ \cite[Subsections 10.21(iv) and 10.21(vii)]{NIST}, 
solutions to \eqref{eq:masslessvac-excitations-mpos} exist if and only if $a\Omega>1$. 
That $a\Omega>1$ is a necessary condition for excitations has been observed in~\cite{RotatingVacuum}. 

We recall that the detector's excitations in the Fock vacuum for $a\Omega>1$ cannot be described in terms of a Bogoliubov transformation between the Fock vacuum and a vacuum co-rotating with the detector, because the latter does not exist: the geometric reason is that points on the boundary $r=a$ that co-rotate with the detector move faster than the speed of light~\cite{RotatingVacuum}. This generalises the older observation about a detector in circular motion in full Minkowski spacetime and the non-existence of a co-rotating vacuum there~\cite{LetawPfautsch}. 
For $a\Omega<1$, by contrast, points on the boundary $r=a$ that co-rotate with the detector move slower than the speed of light, and a co-rotating vacuum exists, but it coincides with the Fock vacuum, consistently with the detector experiencing no excitations~\cite{RotatingVacuum}. 

Finally, we note that when the field is prepared in the thermal state~$\ket{\beta}$, instead of the Fock vacuum, 
a detector undergoes excitations for all values of~$\Omega$, as seen from~\eqref{eq:F vac plus therm}. 
In the limit $\Omega\to0$, the response reduces to a thermal spectrum of excitation and deexcitations, as in full Minkowski space \cite{Takagi,Costa:1994yx, Guimaraes:1998jf,BunneyThermal}.

\subsection{Finite-time interaction}\label{sec: finite time}

In this Section, we write out the passage from the finite-time response function 
$\FF_\chi(\Ebar,\beta)$
\eqref{eq:Fchi+Wbeta}
to the stationary response function 
$\FF(\Ebar,\beta)$
\eqref{eq:F def}, which we analyse in the rest of the paper. 
One reason is to highlight how the delta-peak structure in 
$\FF(\Ebar,\beta,a)$ \eqref{eq:F vac plus therm}
arises in infinite-time limit, whereas the finite time response is continuous under mild regularity conditions on the switching. Another reason is to record the explicit mode-sum formulæ for the finite time response, in the expectation that these formulæ may be applicable to future experimental implementations where the observation time will necessarily be finite. 

Using the stationarity of the derivative correlation function $\WW_\beta$ \eqref{eq:Wbeta-stationary}, the finite time response function 
$\FF_\chi(\Ebar,\beta)$
\eqref{eq:Fchi+Wbeta} can be written as \cite{Fewster,BunneyRick}
\begin{equation}\label{eq:F chi stationary}
    \FF_\chi(\Ebar,\beta)~=~\frac{1}{2\pi}\int_\RR\dd\omega\,\widehat{\WW}_\beta(\omega)|\widehat{\chi}(\Ebar-\omega)|^2\,,
\end{equation}where $\widehat{\WW}_\beta$ and $\widehat{\chi}$ denote the Fourier transforms of $\WW_\beta$ and the switching function $\chi$ respectively, 
in the convention 
\begin{equation}
    \widehat{\WW}_\beta(\omega)~=~\int_{\RR}\dd t\,\ee^{-\ii \omega t}\WW_\beta(t,0)\,,
\end{equation}
and similarly for~$\widehat\chi$. 
Note that $\widehat{\WW}_\beta(\omega)$ is the stationary response function $\FF(\Ebar,\beta)$~\eqref{eq:F def}. 
Substituting the mode-sum expression \eqref{eq:F vac plus therm} for $\widehat{\WW}_\beta(\omega)$ into~\eqref{eq:F chi stationary}, 
we find
\begin{widetext}
    \begin{align}\nonumber
        \FF_\chi(\Ebar,\beta,a)~=~&\frac{1}{2\pi a^2}\sum_{m,n}(\omega_{mn}-m\Omega)^2\frac{J^2_{|m|}(\frac{R}{a}\qmn)}{\omega_{mn}J^2_{|m|+1}(\qmn)}|\widehat{\chi}(\Ebar+\omega_{mn}-m\Omega)|^2\\
        +&\frac{1}{2\pi a^2}\sum_{m,n}(\omega_{mn}-m\Omega)^2n(\beta\omega_{mn})\frac{J^2_{|m|}(\frac{R}{a}\qmn)}{\omega_{mn}J^2_{|m|+1}(\qmn)}|\widehat{\chi}(\Eabs+\omega_{mn}-m\Omega)|^2\nonumber\\
        +&\frac{1}{2\pi a^2}\sum_{m,n}(\omega_{mn}-m\Omega)^2n(\beta\omega_{mn})\frac{J^2_{|m|}(\frac{R}{a}\qmn)}{\omega_{mn}J^2_{|m|+1}(\qmn)}|\widehat{\chi}(\Eabs-\omega_{mn}+m\Omega)|^2\,,\label{eq:finite time response}
    \end{align}
\end{widetext}
where we have again included $a$ as the last argument of $\FF_\chi(\Ebar,\beta,a)$ as a reminder of the dependence on the boundary radius~$a$. The first term in \eqref{eq:finite time response} is the Fock vacuum contribution and the last two terms are the thermal correction. 

The mode-sum formula \eqref{eq:finite time response} 
incorporates 
the Dirichlet condition at $r=a$ in the Bessel functions and their zeros~$\qmn$, 
the field's dispersion relation in 
$\omega_{mn} = \omega(\qmn/a)$, 
the finite temperature in~$\beta$, 
and the finite time effects in~$|\widehat{\chi}|^2$. 
We expect the formula to be applicable to a range of experimental implementations where these pieces of input are determined by the specific analogue spacetime system \cite{Gooding,BunneySounds}. 

The stationary response function \eqref{eq:F def} is obtained from the finite-time response function \eqref{eq:F chi stationary} as the limit in which $|\widehat{\chi}(\omega)|^2\to2\pi\delta(\omega)$. 
This is the limit in which the interaction operates at approximately constant strength for a long time, the total transition probability is divided by the interaction duration, and at the end the limit of infinite duration is taken, yielding the transition probability per unit time. 
The infinite time limit has been implemented in this form since the introduction of UDW detectors~\cite{Unruh,DeWitt:1980hx}, and recent discussions of the uniformity of the limit have been given in \cite{Fewster,BunneyRick,ParryWaiting,ParryNecessity}. 
In the present paper, the key difference between the finite-time response function $\FF_\chi$ \eqref{eq:finite time response} and the stationary response function $\FF$~\eqref{eq:F vac plus therm} is that $\FF_\chi$ is a function, under mild regularity conditions on~$\chi$, 
whereas $\FF$ is a distribution, consisting of Dirac delta peaks. Upon considering the large-boundary limit, the limit $\FF_\chi$ would hence involve a limit of a function, but the limit of $\FF$ involves the limit of a distribution. 
We shall consider the large boundary limit of $\FF$ in Section~\ref{sec:large boundary}, and will hence need to use techniques that accommodate the distributional character of~$\FF$.


\section{Large-boundary limit}\label{sec:large boundary}
In this Section, we find the leading and subleading contributions to the long time thermal response function \eqref{eq:F vac plus therm} in the large-boundary limit $a\to\infty$, with all other parameters fixed. We specialise to a massless Klein-Gordon field, with the dispersion relation $\omega_{mn}=\qmn/a$. We consider first the vacuum contribution \eqref{eq:F vac} and then the thermal corrections~\eqref{eq:F therm}.

\subsection{Vacuum contribution}\label{sec: large a vac sec}

With $\omega_{mn}=\qmn/a$, the vacuum contribution $\FF_\infty(\Ebar,a)$ \eqref{eq:F vac} reads 
\begin{multline}\label{eq:F vac KG}
    \FF_\infty(\Ebar,a)~=~\frac{\Ebar^2}{a}\sum_{m,n}\frac{J^2_{|m|}(\frac Ra\qmn)}{\qmn J^2_{|m|+1}(\qmn)}\\
    \times\delta(\qmn/a-m\Omega+\Ebar)\,. 
\end{multline}
As $\FF_\infty$ \eqref{eq:F vac KG} is a distribution but we expect (and shall find) the leading term in its $a\to\infty$ expansion to be a function, we consider the \textit{integrated response function}, or simply integrated response, defined by 
\begin{equation}\label{eq:G def}
    \GG(a)~=~\int_{\RR}\dd \Ebar\,\Ebar^{-2}\sigma(\Ebar)\FF_{\infty}(\Ebar,a)\,,
\end{equation}where $\sigma\in C^\infty_0(\RR)$ is a real-valued smooth function of compact support, 
such that either $\supp(\sigma)\subset\RR_{>0}$ or $\supp(\sigma)\subset\RR_{<0}$. 
Substituting~\eqref{eq:F vac KG} into~\eqref{eq:G def}, the integrated response reads
\begin{multline}\label{eq:G vac a}
    \GG(a)~=~\frac{1}{a}\sum_{m,n}\frac{J^2_{|m|}(\frac Ra\qmn)}{\qmn J^2_{|m|+1}(\qmn)}\sigma(m\Omega-\qmn/a)\,.
\end{multline} 

We show in Appendix \ref{app: large a vac} that 
\begin{multline}
\GG(a)~=~\int_{\RR}\dd\Ebar\,\sigma(\Ebar)\bigg[\frac12\!\!\!\!\sum_{m>\Ebar/\Omega}\!\!\!\!J^2_{|m|}(mv-\Ebar R)\\
-\frac{R^2}{a^3}\kappa\delta(|\Ebar|-\Omega)+o(a^{-3})\bigg]\,,
\label{eq:G-vacuum-result}
\end{multline}
where $\kappa$ is a numerical constant, given by 
\begin{align}
\nonumber
    \kappa~=~\frac14\bigg\{&\sum_{n=1}^\infty\bigg(\frac{\pi^3}{2}(n+\tfrac14)^2-\frac{3\pi}{16}-\frac{q_{1n}}{J^2_2(q_{1n})}\bigg)\\&+\frac{\pi(5\pi^2-18)}{128}\bigg\}\,, 
\label{eq: kappa}
\end{align}
with the approximate value $\kappa \approx 0.2$. 
From \eqref{eq:G def} we then have 
\begin{multline}\label{eq:F vac expand}
    \FF_\infty(\Ebar,a)~=~\frac{\Ebar^2}{2}\!\!\!\!\sum_{m>\Ebar/\Omega}\!\!\!J^2_{|m|}(mv-\Ebar R)\\
        -\frac{\Ebar^2R^2}{a^3}\kappa\delta(|\Ebar|-\Omega)+o(a^{-3})\,,
\end{multline}
where the $o$-notation is understood in the distributional sense. 

The leading, $a$-independent term in \eqref{eq:F vac expand} is the mode-sum representation of the response function of a detector in Minkowski vacuum in full $2+1$ Minkowski spacetime~\cite{BunneyThermal}. This was to be expected. 

The subleading term in \eqref{eq:F vac expand} is proportional to~$a^{-3}$, and it consists of resonance delta-peaks
at $\Ebar = \pm\Omega$. These peaks occur at the smallest values of $\Eabs$ where the leading term has a discontinuity in its first derivative \cite{Biermann,BunneyThermal,Bunneythesis}. Note that the subleading term is negative, showing that the boundary reduces the density of states; this is what one might have expected on physical grounds. 

As a side outcome, we establish in Appendix~\ref{app: large a vac} the asymptotic expansion
\begin{equation}\label{eq: new asymptotic}
    \frac{1}{\qmn J^2_{|m|+1}(\qmn)}~\sim~\frac{\pi}{2}X'_{|m|}(\pi\amn)\,,
\end{equation}as $n\to\infty$, where $\amn\coloneq n+\tfrac12|m|-\tfrac14$ and $X_{|m|}$ is the McMahon asymptotic expansion for large zeros of the Bessel function~\cite{NIST}. The proof of~\eqref{eq: new asymptotic} was provided by Gerg\H{o} Nemes (Tokyo Metropolitan University) \cite{Nemes-private-2024} and is given in Appendix~\ref{app: asymptotic expansion}. 
The left-hand side of \eqref{eq: new asymptotic} appears for us as the normalisation factor in~\eqref{eq:F vac KG}. We have not encountered the expansion \eqref{eq: new asymptotic} in the existing literature.

\subsection{Thermal correction}\label{sec: large a beta sec}
With $\omega_{mn}=\qmn/a$, the thermal correction~$\Delta\FF_\beta(\Ebar,a)$ \eqref{eq:F therm} reads 
    \begin{align}
    \Delta\FF_\beta(\Ebar,a)&=~\frac{\Ebar^2}{a}\sum_{m,n}n(\tfrac{\beta}{a} \qmn)\frac{J^2_{|m|}(\frac{R}{a}\qmn)}{\qmn J^2_{|m|+1}(\qmn)}\nonumber\\
    \times\big(\delta(\qmn/&a-m\Omega+\Eabs)+\delta(\qmn/a-m\Omega-\Eabs)\big)\,,\label{eq:F beta vac KG}
    \end{align}
where $n$ is the Planckian factor~\eqref{eq:n Planck}.

As with the vacuum contribution, we introduce the integrated response by
\begin{equation}\label{eq:G beta def}
    \Delta\GG_\beta(a)~=~\int_\RR\dd\Ebar\,\Ebar^{-2}\sigma(\Ebar)\Delta\FF_\beta(\Ebar,a)\,,
\end{equation}
where $\sigma\in C^\infty_0(\RR)$ is again a test function with the properties stated below~\eqref{eq:G def}. 
Substituting \eqref{eq:F beta vac KG} into~\eqref{eq:G beta def}, we find
\begin{multline}\label{eqn:G beta a}
    \Delta\GG_\beta(a)~=~\frac 1a\sum_{m,n}n(\tfrac{\beta}{a}\qmn)\frac{J^2_{|m|}(\frac R a\qmn)}{\qmn J^2_{|m|+1}(\qmn)}\\
    \times\bigg(\sigma(m\Omega-\qmn/a)+\sigma(\qmn/a-m\Omega)\bigg)\,.
\end{multline}

We show in Appendix \ref{app:large a therm} that 
\begin{multline}
\label{eq:DeltaGbeta-integral-maintext}
\Delta\GG_\beta(a)~=~\int_\RR\dd\Ebar\,\sigma(\Ebar)\bigg[\frac12\!\!\!\!\sum_{m>|\Ebar|/\Omega}\!\!\!\!n(\beta\omega_+)J^2_{|m|}(\omega_+R)\\+\frac12\!\!\!\!\!\!\sum_{m>-|\Ebar|/\Omega}\!\!\!\!n(\beta\omega_-)J^2_{|m|}(\omega_-R)\\-\frac{R^2}{\beta a^2}\wt{\kappa}\delta(|\Ebar|-\Omega)+o(a^{-2})\bigg]\,,
\end{multline}
where $\omega_\pm=m\Omega\mp\Eabs$, and 
$\wt{\kappa}$ is a numerical constant, given by 
\begin{align}
\nonumber
    \wt{\kappa}~=~\frac12\bigg\{&\sum_{n=1}^\infty\bigg(\frac{\pi^2}{2}(n+\tfrac14)-\frac{1}{J^2_2(q_{1n})}\bigg)\\&+\frac{(23\pi^2-36)}{192}\bigg\}\,,
\label{eq:kappa 2}
\end{align} 
with the approximate value $\wt{\kappa}\approx 0.5$. 
From \eqref{eq:G beta def} we then have
\begin{align}\nonumber
    \Delta\FF_\beta(\Ebar,a)~=~&\frac{\Ebar^2}{2}\!\!\!\!\sum_{m>\Ebar/\Omega}\!\!\!n(\beta\omega_+)J^2_{|m|}(\omega_+R)\\
    \nonumber+&\frac{\Ebar^2}{2}\!\!\!\!\!\!\sum_{m>-\Ebar/\Omega}\!\!\!\!\!n(\beta\omega_-)J^2_{|m|}(\omega_-R)
    \\-&\frac{\Ebar^2R^2}{\beta a^2}\wt{\kappa}\delta(\Eabs-\Omega)+o(a^{-2})\,,\label{eq:F beta expand}
\end{align}where the $o$-notation is again understood in the distributional sense. 

The leading, $a$-independent term in~\eqref{eq:F beta expand} is the mode-sum representation of the thermal correction to the response function in a thermal state in full $2+1$ Minkowski spacetime~\cite{BunneyThermal}. This was again to be expected. 

The subleading term in~\eqref{eq:F beta expand} consists again of resonance delta-peaks at $\Ebar = \pm\Omega$.  Its $a^{-2}$ falloff is however slower than the $a^{-3}$ falloff of the vacuum contribution subleading term. The detector is hence more sensitive to the boundary when an ambient temperature is present.

\subsection{Static detector}\label{sec: staticdetector}

As a consistency check, we consider briefly the static-observer limit, $\Omega\to0$ with $R$ fixed. 

In the vacuum contribution~\eqref{eq:F vac expand}, the $\Omega\to0$ limit of the leading term becomes the Minkowski vacuum response for a static detector, equal to $\frac12 \Ebar^2 \Theta(-\Ebar)$, as can be verified from \eqref{eq:F vac expand} using Neumann's addition formula $\sum_{m\in\ZZ}J^2_{|m|}(x)=1$~\cite[(10.23.3)]{NIST}.
The $\Omega\to0$ limit of the subleading term is proportional to~$\Ebar^2 \delta(\Ebar)$, which vanishes since $\Ebar\neq0$ by assumption. Assuming that the $\Omega\to0$ limit in the expansion can be taken order by order, we have 
\begin{equation}\label{eq:F stat vac}
    \lim_{\Omega\to0}\FF_{\infty}(\Ebar,a)~=~\frac{\Ebar^2}{2}\Theta(-\Ebar)+o(a^{-3})\,.
\end{equation}

In the finite temperature correction~\eqref{eq:F beta expand}, proceeding similarly shows that the $\Omega\to0$ limit of the leading term becomes the thermal correction in full Minkowski space, 
equal to $\frac12 \Ebar^2 n(\beta\Eabs)$~\cite{BunneyThermal}, and the subleading term vanishes. 
Assuming again that the $\Omega\to0$ limit can be taken order by order, we have
\begin{equation}
    \lim_{\Omega\to0}\!\Delta\FF_\beta(\Ebar,a)
    ~=~\frac{\Ebar^2}{2} n(\beta\Eabs)+o(a^{-2})\,.\label{eq:F stat beta}
\end{equation}

Combining~\eqref{eq:F stat vac} and~\eqref{eq:F stat beta}, a static observer in a thermal bath registers the response function
\begin{equation}\label{eq:F stat}
    \FF_{\mrm{static}}(\Ebar,a,\beta)~=~\frac{\Ebar^2}{2}\Theta(-\Ebar)+\frac{\Ebar^2}{2}n(\beta\Eabs)+o(a^{-2})\,.
\end{equation}
We shall not attempt to analyse the finite size correction beyond the bound shown in~\eqref{eq:F stat}, but we note from this bound that a static detector is less sensitive to the presence of the boundary than a detector in circular motion. 
We may also note from \eqref{eq:F vac KG}
and 
\eqref{eq:F beta vac KG} that $\FF_{\mrm{static}}(\Ebar,a,\beta)$ is a sum of Dirac deltas at $\Ebar= \pm \qmn/a$, at the eigenfrequencies of the cavity, and this sum satisfies for every $a$ the detailed balance condition at temperature $1/\beta$, 
\begin{equation}
    \FF_{\mrm{static}}(-\Ebar,a,\beta)~=~\ee^{\beta\Ebar}\FF_{\mrm{static}}(\Ebar,a,\beta)\,.
\end{equation}
These properties of $\FF_{\mrm{static}}$ follow directly from the staticity of the cavity-detector system.

\subsection{Comparison with a detector in anti-de Sitter spacetime}\label{sec: adS comparison}

As another consistency check, we compare the results to those for a detector in circular motion in $(2+1)$-dimensional anti-de Sitter (AdS) spacetime~\cite{Bunneydsads}, where the negative cosmological constant provides a notion of spatial confinement: 
a cosmological constant $\Lambda<0$ corresponds to confinement at the length scale~$1/\sqrt{-\Lambda}$, which scale is thus analogous to our cavity radius~$a$, 
and the limit $a\to\infty$ corresponds to the limit $\Lambda \to 0$. 

In the $\Lambda \to 0$ limit, the response in AdS spacetime reduces to the response in Minkowski spacetime. The leading small-$\Lambda$ corrections were found in~\cite{Bunneydsads}. 
In the AdS counterpart of our vacuum contribution~\eqref{eq:F vac expand}, 
the leading corrections appear in order~${(-\Lambda)}^{3/2}$, and they consist of resonance excitation and de-excitation peaks similar to those in~\eqref{eq:F vac expand}, plus additional curvature corrections. For the vacuum contribution, there is thus a close correspondence between the $a^{-3}$ correction in \eqref{eq:F vac expand} and the AdS ${(-\Lambda)}^{3/2}$ correction. In the AdS counterpart of our thermal correction~\eqref{eq:F beta expand}, 
the leading small-$\Lambda$ correction contains excitation and de-excitation resonance peaks, as in~\eqref{eq:F beta expand}, but these corrections appear in order $\Lambda\ln(-\Lambda)$, which is larger than the order $a^{-2}$ in~\eqref{eq:F beta expand}. 
This indicates that the interplay between 
thermal states and finite size is more subtle with confinement by a negative cosmological constant than with confinement by a boundary in Minkowski spacetime.  We leave further investigation of these subtleties to future work.

\section{Conclusions}\label{sec:conc}
Motivated by proposals in condensed-matter systems to observe the analogue-spacetime circular motion Unruh effect \cite{BEC1,Marino:2020uqj,Gooding,BunneySounds,Gooding:2025tfp}, we addressed the effects of thermality and spatial confinement on a pointlike two-level quantum system undergoing uniform circular motion in ($2+1$)-dimensional analogue Minkowski spacetime with a circular boundary. 

We modelled the pointlike quantum system as an Unruh-DeWitt detector that is coupled linearly to the time derivative of a quantised massless scalar field. The dynamics and the energy gap of the detector were parametrised in terms of the Minkowski time, rather than the detector’s proper time. This setting is expected to more closely connect to the quantities measured in an analogue spacetime experiment~\cite{Gooding,BunneySounds,Gooding:2025tfp}. 
We worked in the limit of weak interaction within first-order perturbation theory and long interaction time, and we neglected the back-action of the detector on the field.

We investigated the large-boundary limit, recovering to leading order the response function of a UDW detector undergoing uniform circular motion in spatially infinite $2+1$ Minkowski spacetime. At subleading order, we found resonance delta-peaks located at the detector's rotational frequency: these peaks coincide with the discontinuities in the first derivative of the response function of a detector undergoing uniform circular motion in $2+1$ Minkowski spacetime~\cite{Biermann,BunneyThermal}. The overall effect of these resonance peaks is a decrease in the detector response, physically corresponding to the reduced density of states expected from a finite-size system. 

These results bear a similarity to the rotating detector's response in anti-de Sitter (AdS) spacetime~\cite{Bunneydsads}, where the negative cosmological constant provides a notion of spatial confinement and the role of the boundary is played by the asymptotically AdS infinity. By comparison with~\eqref{eq:F vac expand}, one may identify whether a contribution is due to spatial curvature or confinement. As such, the techniques developed in this paper have applications not just in flat spacetimes with boundaries but also in curved spacetimes.

As a mathematical side outcome, we found an asymptotic expansion for the Bessel-function normalisation factor \eqref{eq: new asymptotic} in terms of the McMahon expansion of the large zeros of Bessel functions~\cite{NIST}. 
We have not encountered this expansion in the existing literature.

As a mathematical curiosity, we encountered in the analysis two numerical constants, $\kappa \approx 0.2$ \eqref{eq: kappa} and 
$\wt{\kappa} \approx 0.5$ \eqref{eq:kappa 2}, each given by a sum involving zeroes of the Bessel function~$J_1$. 
Our numerical experiments are consistent with the hypothesis that $\kappa$ might equal exactly $\frac15$ and $\wt{\kappa}$ might equal exactly~$\frac12$. We have not attempted to examine this hypothesis analytically. 

We recall that our motivation to consider both a finite size and a finite ambient temperature came from the proposals to simulate a UDW detector in circular motion in a genuinely relativistic spacetime, which use a laser to probe quasiparticles in a condensed-matter system, such as phonons in a Bose-Einstein condensate or third-sound waves in superfluid helium~\cite{Gooding,BunneySounds}. The dynamics of perturbations in these systems simulate a scalar field in a relativistic spacetime when the detector energy gap is sufficiently low and dispersive effects may be neglected. However, these systems can never simulate a field at zero temperature in a spatially infinite spacetime. 

A further experimental consideration is the effect of a finite interaction time and the potential transient effects from the beginning and end of the interaction. In this regard, equation~\eqref{eq:finite time response} is likely to model the true experience of a detector probing these systems. We leave the detailed investigation of finite-time effects under spatial confinement and ambient temperature to future work.

\section*{Acknowledgements}

We thank Gerg\H{o} Nemes (Tokyo Metropolitan University) for discussions on asymptotics of the zeros of Bessel functions and for providing the proof of the asymptotic expansion~\eqref{eq: new asymptotic}, Leo J A Parry and Leonardo Solidoro for helpful comments on an early version of the manuscript, and Silke Weinfurtner and other members of the Nottingham Gravity
Laboratory for numerous helpful interactions. The work of CRDB was supported by United Kingdom Research and Innovation Engineering and Physical Sciences Research Council [EP/W524402/1]. The work of JL was supported by United Kingdom Research and Innovation Science and Technology Facilities Council [grant numbers ST/S002227/1, ST/T006900/1 and ST/Y004523/1].
For the purpose of open access, the authors have applied a CC BY public copyright licence to any Author Accepted Manuscript version arising.

\appendix
\onecolumngrid

\section{Zeros of $J_m$: auxiliary results}\label{app: asymptotic expansion}

In this appendix, we record a version of the McMahon expansion \cite[(10.21.19)]{NIST} for the large zeros of $J_m$ and use this expansion to establish Proposition~\ref{propo:nemes}, 
which will be used in the expansions in Appendices \ref{app: large a vac} and~\ref{app:large a therm}\null. 

Let $m \in \{0,1,2,\ldots\}$, and let $q_{mn}$ denote the positive zeros of the Bessel function $J_m$, $n \in \{1,2,\ldots\}$. 
The McMahon expansion \cite[(10.21.19)]{NIST} is an asymptotic expansion of $q_{mn}$ as $n\to\infty$ with fixed~$m$, given by 
\begin{equation}\label{eqn:m-McMahon expand}
        q_{mn}~\sim~X_{m}(\pi \alpha_{mn})\,,
\end{equation}
where
\begin{subequations}
\label{eqn:alpha-mn}
    \begin{align}
        \alpha_{mn}~\coloneq&~n+\frac12 m-\frac14\,,
        \label{eqn:alphamn-def}\\
        X_{m}(z)~\coloneq&~
        z + \sum_{k=1}^\infty\frac{c_k(m)}{z^{2k-1}}~=~ z -\frac{4m^2-1}{8z}
        +O(z^{-3})\,,
        \label{eqn:X-mz}
    \end{align}
\end{subequations}
the equals sign in \eqref{eqn:X-mz} holds in the sense of the asymptotic expansion~\eqref{eqn:m-McMahon expand}, and the coefficients $c_k(m)$ are even polynomials of degree $2k$ in~$m$. 
$c_1(m)$~is shown in~\eqref{eqn:X-mz}, 
$c_2(m)$, $c_3(m)$ and $c_4(m)$ are given in \cite[(10.21.19)]{NIST}, 
and the higher $c_k(m)$ are obtained from the recurrence relation given in~\cite{NemesZeros}.

\begin{propo}\label{propo:nemes}
For fixed $m \in \{0,1,2,\ldots\}$ and $n\to\infty$, we have 
\begin{equation}\label{eqn:denom expand}
    \frac{1}{q_{mn} J_{m+1}^2(q_{mn})}~\sim~\frac{\pi}{2}X'_{m}(\pi\alpha_{mn}) \,,
\end{equation}
where $\alpha_{mn}$ is given in \eqref{eqn:alphamn-def} and 
$X'_m$ is obtained by differentiating \eqref{eqn:X-mz} term by term, 
\begin{align}
X_{m}'(z)~=&~1 -\sum_{k=1}^\infty\frac{(2k-1) c_k(m)}{z^{2k}}\,, 
\label{eqn:X-mz-prime}
\end{align}
where the equals sign holds in the sense of the asymptotic expansion~\eqref{eqn:denom expand}. 
\end{propo}

We have not encountered Proposition \ref{propo:nemes} 
in the existing literature. 
The following proof was provided by Gerg\H{o} Nemes (Tokyo Metropolitan University)~\cite{Nemes-private-2024}.
\begin{proof}
In the notation of \cite[$\S$10.21(ii)]{NIST}, let 
$C_{m}(z,t)\coloneq J_{m}(z)\cos(\pi t)+Y_{m}(z)\sin(\pi t)$, 
where 
$Y_{m}$ are the Bessel functions of the second kind, 
and we have explicitly indicated the dependence of $C_{m}(z,t)$ on the parameter $t\ge0$. 
Then, there exists a continuous function $\rho_m(t)$ that satisfies $\rho_m(0)=0$, $C_{m}(\rho_m(t),t) = 0$ for $t>0$, and $\rho_m(n) = q_{mn}$, $n \in \{1,2,\ldots\}$. 

Using the recurrence relation~\cite[(10.6.2)]{NIST}, we have $J_{m}'(q_{mn})=-J_{m+1}(q_{mn})$. As such, we may write
\begin{equation}\label{eqn: jm qmn rho n}
    q_{mn} J^2_{m+1}(q_{mn})~=~\lrl{\lrb{\rho_{m}(t)\lr{J'_{m}(\rho_{m}(t))}^2}}{t=n}\,.
\end{equation}By~\cite[(10.21.10)]{NIST}, we have
\begin{equation}\label{eqn: jm qmn rho}
    J'_{m}(\rho_{m}(t))~=~{\lrb{\frac{1}{2}\rho_{m}(t)\td{\rho_{m}(t)}{t}}}^{-\frac{1}{2}}\,.
\end{equation}
Hence 
\begin{equation}\label{eqn:exact equal}
    \frac{1}{q_{mn} J^2_{m+1}(q_{mn})}~=~\lrl{\frac{1}{\rho_{m}(t)\lr{J_{m}'\lr{\rho_{m}(t)}}^2}}{t=n}~=~\lrl{\frac{1}{2}\td{\rho_{m}(t)}{t}}{t=n}\,,
\end{equation}
where the first equality uses \eqref{eqn: jm qmn rho n} and the second equality uses~\eqref{eqn: jm qmn rho}. 

The large $t$ expansion of $\rho_{m}(t)$ is obtained from \eqref{eqn:m-McMahon expand} and \eqref{eqn:alpha-mn} by replacing $\alpha_{mn}$ by 
$t+\frac12m-\frac14$ \cite[comment below (10.21.19)]{NIST}. Combining this with~\eqref{eqn:exact equal}, 
\eqref{eqn:denom expand} follows. 
\end{proof}

Note that the series for 
$X_{m}$ \eqref{eqn:X-mz}
and 
$X'_{m}$
\eqref{eqn:X-mz-prime}
hold in the sense of the asymptotic expansions 
\eqref{eqn:m-McMahon expand}
and 
\eqref{eqn:denom expand}
but are not convergent. 
Let ${\hat X}_{m}$
and 
${\hat X}'_{m}$ be the truncations of these series to 
$1 \le k \le k_\text{tr}$, where the notation suppresses the value of 
$k_\text{tr} \in \{1,2,\ldots\}$. 
${\hat X}_{m}$
and 
${\hat X}'_{m}$ with sufficiently high $k_\text{tr}$ will be used in Appendices \ref{app: large a vac} and~\ref{app:large a therm}\null. 

\section{Large-$a$ asymptotics: vacuum contribution}\label{app: large a vac}

In this appendix, we find the asymptotic behaviour of the vacuum integrated response function \eqref{eq:G vac a} in the large-$a$ regime, leading to the results of Section~\ref{sec: large a vac sec}.

\subsection{Decomposition of integrated response function}
The integrated response function \eqref{eq:G vac a} reads
\begin{equation}\label{eqn:G vac a}
    \GG(a)~=~\frac{1}{a}\sum_{m,n}\frac{J^2_{|m|}\lr{\frac{R}{a}\qmn}}{\qmn J_{|m|+1}^2(\qmn)}\sigma(m\Omega-\qmn/a)\,,
\end{equation}
where the sum is over $m\in\ZZ$ and $n\in\NN$. 

We recall that $\sigma\in C_0^\infty(\RR)$. Let $\sigma_I\coloneq \inf\supp(\sigma)$ and $\sigma_S\coloneq\sup\supp(\sigma)$. We assume that the support of $\sigma$ is chosen such that either $\supp(\sigma)\subset\RR_{>0}$ with $0<\sigma_I<\Omega<\sigma_S$ or $\supp(\sigma)\subset\RR_{<0}$ with $\sigma_I<-\Omega<\sigma_S<0$.

For fixed $m$ in~\eqref{eqn:G vac a}, the values of $n$ in the sum over $n$ are restricted by
\begin{equation}
    \sigma_I~<~m\Omega-\frac{\qmn}{a}~<~\sigma_S\,,
\label{eq:qmn-rawbounds}
\end{equation}
which may be rewritten as
\begin{equation}\label{eqn: qmn inequal}
m-\frac{\sigma_S}{\Omega}
~<~\frac{\qmn}{a\Omega}~<~m -\frac{\sigma_I}{\Omega}\,.
\end{equation} 
Depending on the value of~$m$, the condition~\eqref{eqn: qmn inequal} falls into three different cases. For sufficiently large~$a$, with
the other parameters fixed, these cases are as follows.
\begin{itemize}
        \item[-] For $m \leq \frac{\sigma_I}{\Omega}$, no $n$ satisfy~\eqref{eqn: qmn inequal}.
        \item[-] For $\frac{\sigma_I}{\Omega}<m\leq\frac{\sigma_S}{\Omega}$, $n$ satisfies $1\leq n\leq\Nmax$. 
        We denote the set of these $m$ by $\mc{C}^{1}$. 
        \item[-] For $\frac{\sigma_S}{\Omega} < m$, $n$ satisfies $\Nmin\leq n\leq\Nmax$. 
        We denote the set of these $m$ by $\mc{C}^{2}$. 
    \end{itemize}
Here $\Nmin$, $\Nmax$ are respectively the least and greatest values of $n$ satisfying~\eqref{eqn: qmn inequal}, and the notation suppresses their dependence on $a$ and~$m$. 
Note that $\mc{C}^1$ is a finite set, and it contains $1$ when 
$0<\sigma_I<\Omega<\sigma_S$ 
and $-1$ when $\sigma_I<-\Omega<\sigma_S<0$. 

With this notation, we may split $\GG$ as
\begin{subequations}
\label{eq:G-split}
    \begin{align}
        \GG(a)&~=~\GG^1(a)+\GG^2(a)\,,\\
        \GG^1(a)&~=~\frac{1}{a}\sum_{m\in\mc{C}^1}\sum_{n=1}^{\Nmax}\frac{J^2_{|m|}\lr{\frac{R}{a}\qmn}}{\qmn J_{|m|+1}^2(\qmn)}\sigma(m\Omega-\qmn/a)\,,\label{eqn:G1 a}\\
        \GG^2(a)&~=~\frac{1}{a}\sum_{m\in\mc{C}^2}\sum_{n=\Nmin}^{\Nmax}\frac{J^2_{|m|}\lr{\frac{R}{a}\qmn}}{\qmn J_{|m|+1}^2(\qmn)}\sigma(m\Omega-\qmn/a)\,.\label{eqn:G2 a}
    \end{align}
\end{subequations}

\subsection{$\GG^{1}$}
\label{appsec:GG1}

We consider first $\GG^1$~\eqref{eqn:G1 a}. To further decompose the sum over~$n$, we fix a constant $p\in(0,\tfrac14)$ and set $N\coloneq\lfloor\lr{\frac aR}^p\rfloor$, where $\lfloor\cdot\rfloor$ is the floor function~\cite{NIST}. For sufficiently large $a$, we then have $N<\Nmax$, and we may write
\begin{subequations}
    \begin{align}
        \GG^1(a)&~=~\GG_<^1(a)+\GG_>^1(a)\,\\\label{eqn:g1 less}
        \GG_<^1(a)&~=~\frac{1}{a}\sum_{m\in\mc{C}^1}\sum_{n=1}^{N-1}\frac{J^2_{|m|}\lr{\frac{R}{a}\qmn}}{\qmn J_{|m|+1}^2(\qmn)}\sigma(m\Omega-\qmn/a)\,,\\\label{eqn:g1 greater}
        \GG_>^1(a)&~=~\frac{1}{a}\sum_{m\in\mc{C}^1}\sum_{n=N}^{\Nmax}\frac{J^2_{|m|}\lr{\frac{R}{a}\qmn}}{\qmn J_{|m|+1}^2(\qmn)}\sigma(m\Omega-\qmn/a)\,.
    \end{align}
\end{subequations}

To address $\GG_<^1$~\eqref{eqn:g1 less}, we recall that $\mc{C}^1$ has finitely many elements and does not contain $m=0$, and as $n<N$ in~\eqref{eqn:g1 less}, the McMahon expansion~\eqref{eqn:m-McMahon expand} with $m$ replaced by $|m|$ informs us that $\qmn\sim \pi n+O(1)$, which is at most $q_{|m|N}=O(a^p)$. As such for $n<N$, we have $\qmn/a=o(1)$ as $a\to\infty$. Then, elementary estimates, using the $m\ne0$ Maclaurin expansion~\cite{NIST}
\begin{equation}\label{eqn:small Bessel}
    J^2_{|m|}(z)~\sim~\frac14 z^2\delta_{|m|1}+O(z^4)\,,
\end{equation}as $z\to0$, and recalling $0<p<\tfrac14$, show
\begin{equation}\label{eqn:G1 minus}
    \GG^1_<(a)~=~\frac{R^2}{4a}\sum_{m\in\mc{C}^1}\sum_{n=1}^{N-1}\frac{1}{q_{1n}J_2^2(q_{1n})}\delta_{1|m|}\lr{\lr{\frac{q_{1n}}{a}}^2\sigma(m\Omega)-\lr{\frac{q_{1n}}{a}}^3\sigma'(m\Omega)}+o(a^{-3})\,,
\end{equation}where $\sigma'(m\Omega)=\frac{\dd}{\dd z}\sigma(z)|_{z=m\Omega}$. 

By Proposition~\ref{propo:nemes}, the factor $1/(q_{1n}J_2^2(q_{1n}))$ in \eqref{eqn:G1 minus} is bounded, and by the McMahon expansion~\eqref{eqn:m-McMahon expand}, $q_{1n}$ is bounded by a multiple of~$n$. The second term under the sum in~\eqref{eqn:G1 minus} is hence bounded by a multiple of~$n^3/a^3$. Summing over~$n$, the contribution from the second term is hence bounded by a multiple of $N^4/a^4$, which is $o(a^{-3})$ since $N = O(a^p)$ and $0<p<\tfrac14$. We hence have 
\begin{equation}\label{eqn:G1 minus-improved}
    \GG^1_<(a)~=~\frac{R^2}{4a^3}\sum_{m\in\mc{C}^1}\sigma(m\Omega)\delta_{1|m|}\sum_{n=1}^{N-1}\frac{q_{1n}^2}{q_{1n}J_2^2(q_{1n})}+o(a^{-3})\,.
\end{equation}

To address $\GG_>^1$~\eqref{eqn:g1 greater}, 
we recall again that $\mc{C}^1$ has finitely many elements and does not contain $m=0$, 
and we note that $n$ may be considered large in each term in the sum. We may therefore again employ the McMahon expansion \eqref{eqn:m-McMahon expand} for~$\qmn$. 
Using Proposition~\ref{propo:nemes} for the denominator in~\eqref{eqn:g1 greater}, 
and recalling that the McMahon expansion~\eqref{eqn:m-McMahon expand} proceeds in inverse integer powers of~$n$, 
we may write $\GG_>^1$ \eqref{eqn:g1 greater} as 
\begin{equation}\label{eqn:G1 max expand}
    \GG^1_>(a)~=~\frac{\pi}{2a}\sum_{m\in\mc{C}^1}\sum_{n=N}^{\Nmax}J^2_{|m|}\lr{\tfrac Ra {\hat X}_{|m|}(\pi\amn)}{\hat X}'_{|m|}(\pi\amn)\sigma\lr{m\Omega-\tfrac1a {\hat X}_{|m|}(\pi\amn)}
    + {\hat O}^{\infty}(a^{-1})\,, 
\end{equation}
where ${\hat X}$ and ${\hat X}'$ are the truncations of the asymptotic series 
\eqref{eqn:X-mz}
and 
\eqref{eqn:X-mz-prime}
as defined in Appendix~\ref{app: asymptotic expansion}, with the order of truncation suppressed but taken sufficiently high for the steps that follow. The error term ${\hat O}^{\infty}(a^{-1})$ in \eqref{eqn:G1 max expand} can be made to vanish faster than any given inverse power of $a$ as $a\to\infty$ by taking the order of truncation sufficiently high. 

To estimate the sum in~\eqref{eqn:G1 max expand}, we employ the Euler-Maclaurin formula~\cite{Kacbook},
\begin{equation}\label{eqn:Euler Mac}
    \sum_{n=p}^{q}g(n)~=~\int_{p}^q\dd x\,g(x)+\frac12\lr{g(p)+g(q)}+\sum_{i=2}^l\frac{B_i}{i!}\lr{g^{(i-1)}(q)-g^{(i-1)}(p)}+\int_p^q\dd x\,\frac{\wt{B}_l(1-x)}{l!}g^{(l)}(x)\,,
\end{equation}where $l\ge2$ is an integer, $B_i$ are the Bernoulli numbers, and $\wt{B}_l$ are the periodic Bernoulli polynomials~\cite{NIST}.
We first observe that for sufficiently large~$a$, we may replace $\Nmax$ in \eqref{eqn:G1 max expand} by $\Nmax+1$ without changing the value of the sum because the new term is outside the support of~$\sigma$. Applying \eqref{eqn:Euler Mac} with $l=3$, and recalling that $B_2 = \tfrac16$ and $B_3=0$, then obtain
\begin{subequations}
\label{eqn:I-collected}
\begin{align}
    \GG_>^1(a)&~=~I_1+I_2+I_3+I_4 + {\hat O}^{\infty}(a^{-1})\,,\\\label{eqn:I1}
    I_1&~=~\frac{\pi}{2a}\sum_{m\in\mc{C}^1}\int_{N}^{\infty}\dd x\,J^2_{|m|}\lr{\tfrac Ra {\hat X}_{|m|}(\pi\amx)}{\hat X}'_{|m|}(\pi\amx) \sigma\lr{m\Omega-\tfrac1a {\hat X}_{|m|}(\pi\amx)}\,,\\\label{eqn:I2}
    I_2&~=~\frac{\pi}{2a}\sum_{m\in\mc{C}^1}J^2_{|m|}\lr{\tfrac Ra {\hat X}_{|m|}(\pi\amN)} {\hat X}'_{|m|}(\pi\amN)\sigma\lr{m\Omega-\tfrac1a {\hat X}_{|m|}(\pi\amN)}\,,\\\label{eqn:I3}
    I_3&~=~-\frac{\pi}{24a}\sum_{m\in\mc{C}^1}\td{}{x}\bigg(J^2_{|m|}\lr{\tfrac Ra {\hat X}_{|m|}(\pi\amx)} {\hat X}'_{|m|}(\pi\amx)\lrl{\sigma\lr{m\Omega-\tfrac1a {\hat X}_{|m|}(\pi\amx)}\bigg)}{x=N}\,,\\\label{eqn:I4}
    I_4&~=~\frac{\pi}{12a}\sum_{m\in\mathcal{C}^1}\int_{N}^{\infty}\dd x\,\wt{B}_3(1-x)\tdn{}{x}{3}\bigg(J^2_{|m|}\lr{\tfrac Ra {\hat X}_{|m|}(\pi\amx)} {\hat X}'_{|m|}(\pi\amx)\sigma\lr{m\Omega-\tfrac1a {\hat X}_{|m|}(\pi\amx)}\bigg)\,,
\end{align}
\end{subequations}
where $\amx=x+\tfrac12|m|-\tfrac14$, 
the $q=\Nmax+1$ substitution terms in~\eqref{eqn:Euler Mac} have not contributed because of the support of~$\sigma$, and in $I_1$ \eqref{eqn:I1} and $I_4$ \eqref{eqn:I4} we have extended the upper limit of integration to infinity, again because of the support of~$\sigma$. 

For $I_2$ \eqref{eqn:I2} and $I_3$~\eqref{eqn:I3}, we use 
$N/a=o(1)$ to perform a Maclaurin expansion, recalling that $m\ne0$, and we truncate the results at order $o(a^{-3})$. We find 
\begin{subequations}\label{eqn:I23bound}
    \begin{align}
        I_2~=~&\frac{R^2}{16a^3}\sum_{m\in\mc{C}^1}\lr{\pi^3\lr{N+\tfrac14}^2-\tfrac38\pi}\sigma(m\Omega)\delta_{|m|1}+o(a^{-3})\,,\\
        I_3~=~&-\frac{R^2}{48a^3}\sum_{m\in\mc{C}^1}\pi^3\lr{N+\tfrac14}\sigma(m\Omega)\delta_{|m|1}+o(a^{-3})\,.
    \end{align}
\end{subequations}

For $I_1$~\eqref{eqn:I1}, we write $I_1 = I_0 + \Delta I_1$, where $I_0$ is as in \eqref{eqn:I1} but integrated from $0$ to~$\infty$, and the correction $\Delta I_1$ is integrated from $0$ to $N$ and has an overall minus sign. In~$I_0$, we perform the change of variables $z={\hat X}_{|m|}(\pi\amx)/a$. In~$\Delta I_1$, we use again 
$N/a=o(1)$ to perform a Maclaurin expansion, recalling $m\ne0$, and perform the resulting integrals, truncating at order $o(a^{-3})$. We find 
\begin{subequations}\label{eqn:I1bound}
    \begin{align}
        I_1~=~&I_0+\Delta I_1\,,\\
        I_0~=~&\sum_{m\in\mc{C}^1}\int_0^\infty\dd z\,\frac12J_{|m|}^2(Rz)\sigma(m\Omega- z)\,,\\
        \Delta I_1~=~&-\frac{R^2}{24 a^3}\sum_{m\in\mc{C}^1}\lr{\pi^3\lr{N+\tfrac14}^3-\tfrac98\pi\lr{N+\tfrac14}}\sigma(m\Omega)\delta_{|m|1}+o(a^{-3})\,.
    \end{align}
\end{subequations}

For $I_4$~\eqref{eqn:I4}, we perform the change of variables $z=(\pi/a)\amx = (\pi/a)(x+\tfrac12|m|-\tfrac14)$, and recall from \eqref{eqn:X-mz} 
that since $N$ is large, 
we have ${\hat X}_{|m|}(\pi\amx)=\pi\amx+ O(\amx^{-1}) \sim az$ and ${\hat X}'_{|m|}(\pi\amx)\sim 1$. We hence have 
\begin{equation}
    I_4~\sim~\frac{1}{12}\lr{\frac{\pi}{a}}^3\sum_{m\in\mc{C}^1}\int_{0}^\infty\dd z\,\wt{B}_3(\tfrac{a}{\pi}z-\amo)\tdn{}{z}{3}\lr{J^2_{|m|}(Rz)\sigma(m\Omega-z)}\,,
\label{eq:I4-est-N}
\end{equation}
where we have extended the lower limit of the integral from $\frac{\pi}{a} \alpha_{|m| N}$ to $0$ at the expense of an $O(a^{p-4}) = o(a^{-3})$ error. 
As $\tdn{}{z}{3}J_{|m|}^2(Rs)\sigma(m\Omega-z)$ is a bounded function of compact support, a generalisation of the Riemann-Lebesgue lemma that covers the periodic Bernoulli polynomials~\cite[(Theorem 4)]{Riemann} shows that $I_4 = o(a^{-3})$. 

Collecting \eqref{eqn:I23bound} and 
\eqref{eqn:I1bound} and using $I_4 = o(a^{-3})$, we find
\begin{align}\nonumber
    \GG_>^1(a)~=~&\sum_{m\in\mc{C}^1}\int_0^\infty\dd z\,\frac12J_{|m|}^2(Rz)\sigma(m\Omega-z)\\
    &-\frac{R^2}{4a^3}\sum_{m\in\mc{C}^1}\lr{\sum_{n=1}^{N-1}\lrb{\frac{\pi^3}{2}\lr{n+\tfrac14}^2-\frac{3\pi}{16}}+\frac{\pi(5\pi^2-18)}{128}}\sigma(m\Omega)\delta_{|m|1}+o(a^{-3})\,,\label{eqn:G1 plus}
\end{align}where on the last line we have rewritten the cubic polynomial in $N$ as a sum of a quadratic polynomial in~$n$, for reasons that will emerge below in~\eqref{eqn:resum}. 

Combining $\GG^1_<$~\eqref{eqn:G1 minus-improved} and $\GG^1_>$~\eqref{eqn:G1 plus}, we have
\begin{align}
    \GG^1(a)~=~&\sum_{m\in\mc{C}^1}\int_0^\infty\dd z\,\frac12J_{|m|}^2(Rz)\sigma(m\Omega-z)
    \nonumber\\
    &-\frac{R^2}{4a^3}\sum_{m\in\mc{C}^1}\lr{\sum_{n=1}^{N-1}\lrb{\frac{\pi^3}{2}\lr{n+\tfrac14}^2-\frac{3\pi}{16}-\frac{q_{1n}^2}{q_{1n}J_2^2(q_{1n})}}+\frac{\pi(5\pi^2-18)}{128}}\sigma(m\Omega)\delta_{|m|1}
    +\,o(a^{-3})\,.\label{eqn:G1 expand still a}
\end{align}

The expansion \eqref{eqn:G1 expand still a} for $\GG^1(a)$ still contains the auxiliary function $N=\lfloor\lr{\tfrac{a}{R}}^p\rfloor$, with the parameter $p\in(0,\tfrac14)$. To remove the dependence on~$p$, we note that 
\begin{equation}\label{eqn: expansion J frac}
    \frac{q_{1n}^2}{q_{1n}J^2_2(q_{1n})}~=~\frac{\pi^3}{2}\lr{n+\tfrac{1}{4}}^2-\frac{3\pi}{16}+O(n^{-2})\,, 
\end{equation}
using \eqref{eqn:m-McMahon expand} and 
Proposition~\ref{propo:nemes}. In the coefficient of $\sigma(m\Omega)$ in~\eqref{eqn:G1 expand still a}, we may therefore write
\begin{align}
    \sum_{n=1}^{N-1}\lrb{\frac{\pi^3}{2}\lr{n+\tfrac14}^2-\frac{3\pi}{16}-\frac{q_{1n}^2}{q_{1n}J_2^2(q_{1n})}}~=~\sum_{n=1}^{\infty}\lrb{\frac{\pi^3}{2}\lr{n+\tfrac14}^2-\frac{3\pi}{16}-\frac{q_{1n}^2}{q_{1n}J_2^2(q_{1n})}} \ + O(a^{-p})\,.\label{eqn:resum}
\end{align}
The final result for $\GG^1(a)$ is hence 
\begin{align}
    \GG^1(a)~=~&\sum_{m\in\mc{C}^1}\int_0^\infty\dd z\,\frac12J_{|m|}^2(Rz)\sigma(m\Omega-z)\nonumber\\
    &-\frac{R^2}{4a^3}\sum_{m\in\mc{C}^1}\lr{\sum_{n=1}^{\infty}\lrb{\frac{\pi^3}{2}\lr{n+\tfrac14}^2-\frac{3\pi}{16}-\frac{q_{1n}^2}{q_{1n}J_2^2(q_{1n})}}+\frac{\pi(5\pi^2-18)}{128}}\sigma(m\Omega)\delta_{|m|1}
    +\,o(a^{-3})\,,\label{eqn:G1 expand}
\end{align}
showing the $a$-independent leading term and the leading correction, proportional to~$a^{-3}$.

\subsection{$\GG^2$}\label{sec:G a}

We consider next $\GG^2$~\eqref{eqn:G2 a}. As the set $\mc{C}^{2}$ is not bounded above, we will need to use the uniform large argument expansions of the Bessel functions and their zeroes \cite[$\S$10.20]{NIST}. 

To begin, we set $M\coloneq\lfloor(a\Omega)^2\rfloor$, 
and we assume $a$ to be so large that $M\ge1$ and $M > 1 + \sigma_S/\Omega$. We may then decompose the sum over $m$ in \eqref{eqn:G2 a} as 
\begin{subequations}
\label{eqn:G2decomp}
    \begin{align}
        \GG^2(a)&~=~\GG^2_<(a)+\GG^2_>(a)\,,\\\label{eqn:G2min}
        \GG^2_<(a)&~=~\frac{1}{a}\sum_{\substack{m\in\mc{C}^{2}\!\!,\\ m\leq M-1}}\sum_{n=\Nmin}^{\Nmax}\frac{J^2_{|m|}\lr{\frac Ra\qmn}}{\qmn J^2_{|m|+1}(\qmn)}\sigma(m\Omega-\qmn/a)\,,\\
        \GG^2_>(a)&~=~\frac{1}{a}\sum_{m=M}^\infty\sum_{n=\Nmin}^{\Nmax}\frac{J^2_{m}\lr{\frac Raq_{mn}}}{q_{mn} J^2_{m+1}(q_{mn})}\sigma(m\Omega-q_{mn}/a)\,,\label{eqn:G2max}
    \end{align}
\end{subequations}
where we recall that $\Nmin$ and $\Nmax$ depend on $m$ and $a$ as explained below~\eqref{eqn: qmn inequal}. Note that the set of $m$ in \eqref{eqn:G2min} is nonempty because $M > 1 + \sigma_S/\Omega$, and in \eqref{eqn:G2max} we have dropped the absolute values on $m$ because $M\ge1$ and hence $m\ge1$. 

Consider first $\GG^2_>$~\eqref{eqn:G2max}. 
Since $m\ge M \ge1$, the condition \eqref{eqn: qmn inequal} that determines $\Nmin$ and $\Nmax$ can be written as 
\begin{equation}
\label{eqn: qmn inequal alt1}
a \Omega \! \left(
1-\frac{(\sigma_S/\Omega)}{m}
\right)
~<~
\frac{q_{mn}}{m}
~<~
a\Omega \! \left(
1-\frac{(\sigma_I/\Omega)}{m}
\right)\,.
\end{equation} 
As $m\to\infty$, we have from \cite[(10.21.41)]{NIST} the expansion 
\begin{align}
\frac{q_{mn}}{m} = 
z_{mn} \! \left[1 + \frac{1}{24m^2}
\! \left(
\frac{3z_{mn}^2+2}{{(z_{mn}^2-1)}^2} - \frac{5}{3 w_{mn} \sqrt{z_{mn}^2-1}}
\right)
\right]
+ O\!\left(m^{-4}\right)
\,,
\label{qmn/m-asymptotics}
\end{align}
where the $O\!\left(m^{-4}\right)$ term is uniform in~$n$, $z_{mn}>1$ is the solution to 
\begin{align}
w_{mn} &= \sqrt{z_{mn}^2-1} - \arccos\left(\frac{1}{z_{mn}}\right)\,, 
\end{align}  
and 
\begin{align}
w_{mn} &:= \frac{2}{3}\frac{(-a_n)^{3/2}}{m}
\,,
\label{eq:wmn-def}
\end{align}
where 
$a_n$ are the zeroes of the Airy function~$\Ai$, in the notation of \cite[\S9.9]{NIST}. 
From the large $n$ expansion of $a_n$ \cite[(9.9.6)]{NIST} it then follows that as $a\to\infty$, $\Nmin$ and $\Nmax$ are asymptotic to $(a\Omega)m/\pi$. 

To control the denominator in \eqref{eqn:G2max} as $a\to\infty$, we need an expansion that for each $a$ holds for all $m\ge M = \lfloor(a\Omega)^2\rfloor$ and for all $n \in [\Nmin,\Nmax]$. To this end, we write $n = sm + p$, where $s$ and $p$ are integers whose range as $a\to\infty$ is to be determined. 
To implement the lower limit of~$m$, we introduce a formal expansion parameter $\mu>0$, which encodes the largeness of $a\Omega$: in \eqref{qmn/m-asymptotics}, we replace $m$ by $\mu^2 m$ and write $n = \mu^3 sm + \mu p$, and we replace the leftmost and rightmost expressions in \eqref{eqn: qmn inequal alt1} by 
$\mu a \Omega \! \left(
1-\frac{(\sigma_{S,I}/\Omega)}{\mu^2 m}
\right)$. 
We then expand the replaced form of \eqref{eqn: qmn inequal alt1} in $\mu$ as $\mu\to\infty$, and at the end set $\mu=1$.
This shows that the allowed values of $s$ and $p$ have the large $a$ expansions 
\begin{subequations}
\begin{align}
s & = 
\frac{a\Omega}{\pi} - \frac12  
+ 
\frac{1}{2 \pi a\Omega} 
+ 
O \bigl({(a\Omega)}^{-3} \bigr)
\,,
\\
p_{\text{min},\text{max}} & = 
\frac{-(\sigma_{S,I}/\Omega)a\Omega}{\pi} 
+ 
\frac14
+ 
\frac{(\sigma_{S,I}/\Omega)}{2 \pi a\Omega} 
+ 
O \bigl({(a\Omega)}^{-3} \bigr)
\,. 
\end{align}
\end{subequations}
In the denominator in \eqref{eqn:G2max}, we then have, using \cite[(10.20.4) and (9.7.9)]{NIST},
\begin{align}
q_{mn} J^2_{m+1}(q_{mn})
&= 
q_{mn} J^2_{m+1}\left((m+1)
\left(\frac{q_{mn}}{m+1}\right)\right)
\notag 
\\
&\sim 
\frac{2}{\pi}
\cos^2 \! \left[
(m+1) \left(
\sqrt{\left(\frac{q_{mn}}{m+1}\right)^2 -1} 
- \arccos \! \left(\frac{m+1}{q_{mn}}\right)\right)
- \frac{\pi}{4}
\right]
\,. 
\label{eq:qJ^2q-as}
\end{align}
The argument of $\cos^2$ in \eqref{eq:qJ^2q-as} may be shown to be 
$\pi (s m + p - 1) + O \bigl({(a\Omega)}^{-1} \bigr)$ as $a\to\infty$, by introducing the formal expansion parameter $\mu$ as above, expanding as $\mu\to\infty$, and finally setting $\mu=1$. 
Therefore $q_{mn} J^2_{m+1}(q_{mn}) \sim 2/\pi$. 
Substituting this in \eqref{eqn:G2max} gives 
\begin{equation}
    \GG_>^2(a)~\sim~\frac{\pi}{2a}\sum_{m=M}^\infty\sum_{n=\Nmin}^{\Nmax}J^2_{m}\lr{\frac Raq_{mn}}\sigma(m\Omega-q_{mn}/a)\,.
\label{eqn:G2max-i1}
\end{equation}

To bound~\eqref{eqn:G2max-i1}, we recall from \cite[(10.20.4) and (9.7.5)]{NIST}  
that 
\begin{equation}
J^2_{\nu}\lr{\nu z}~\sim~\frac{1}{2\pi\sqrt{1-z^2}}\frac{\ee^{-2 \nu\xi(z)}}{\nu}
\label{eq:Jsq-uniform-as-smallz}
\end{equation} 
as $\nu\to\infty$, valid uniformly for $z\in (0,b]$ for any $b\in (0,1)$, where $\xi(z)\coloneq\ln \! \left(\frac{1+\sqrt{1-z^2}}{z}\right)-\sqrt{1-z^2}>0$. 
In the argument of $J_m$ in~\eqref{eqn:G2max-i1}, \eqref{eqn: qmn inequal alt1} gives $\frac Raq_{mn} = m v \left(1 + O(m^{-1})\right)$, where we recall that $v = R\Omega < 1$. Applying \eqref{eq:Jsq-uniform-as-smallz}, we therefore have 
\begin{equation}
J^2_{m}\lr{\frac Raq_{mn}}~\sim~\frac{1}{2\pi\sqrt{1-v^2}}\frac{\ee^{-2 m\xi(v)}}{m}
\end{equation} 
as $m\to\infty$. 
As $\sigma$ is smooth with compact support, it is bounded, and hence for sufficiently large $a$ we have 
\begin{align}
\GG_>^2(a)~\le~\frac{A}{4a\sqrt{1-v^2}}\sum_{m=M}^\infty\sum_{n=\Nmin}^{\Nmax}\frac{\ee^{-2m \xi(v)}}{m} \,,
\label{eqn:G2max-i2}
\end{align}
where $A$ is a positive constant. 
The sum over $n$ in \eqref{eqn:G2max-i2} gives the factor $\Nmax-\Nmin+1$, and for sufficiently large $a$ this factor is bounded by $(a\Omega)mB$, for some positive constant~$B$, by the observations made below~\eqref{eq:wmn-def}. Hence, for sufficiently large~$a$, we find
\begin{equation}\label{eq:G2 est}
\GG_>^2(a)~\le~\frac{AB\Omega}{4\sqrt{1-v^2}}\sum_{m=M}^\infty\ee^{-2m \xi(v)}~=~\frac{AB\Omega}{4\sqrt{1-v^2}}\frac{\ee^{-2M\xi(v)}}{1-\ee^{-2\xi(v)}}~=~O\lr{\ee^{-2M\xi(v)}}
\,, 
\end{equation}
where we recall that that $M = \lfloor(a\Omega)^2\rfloor$.  

Consider next $\GG^2_<$~\eqref{eqn:G2min}. The range of $m$ in the sum is $\sigma_S/\Omega < m \le M-1$, where we recall that 
$M= \lfloor(a\Omega)^2\rfloor$, and $a$ is assumed so large that $M\ge1$ and $M > 1 + \sigma_S/\Omega$. As $a\to\infty$, the task hence is to control values of $m$ from an $a$-independent lower limit, which may have either sign, to an upper limit that is asymptotic to~$(a\Omega)^2$. 

For fixed~$m$, using in \eqref{eq:qmn-rawbounds} the McMahon expansion \eqref{eqn:m-McMahon expand} with $m$ replaced by $|m|$ shows that values of $n$ in $\GG^2_<$~\eqref{eqn:G2min} grow proportionally to~$a$. Further, for all $m$ in $\GG^2_<$~\eqref{eqn:G2min}, it is seen that the $k$th terms in \eqref{eqn:X-mz}
and 
\eqref{eqn:X-mz-prime} are respectively $O(a^{1-2k})$ and $O(a^{-2k})$. We may therefore write 
\begin{align}
\GG^2_<(a)~=~\frac{\pi}{2a}\sum_{\substack{m\in\mc{C}^{2}\!\!,\\ m\leq M-1}}\sum_{n=\Nmin}^{\Nmax}J^2_{|m|}\lr{\tfrac Ra {\hat X}_{|m|}(\pi\amn)}{\hat X}'_{|m|}(\pi\amn)\sigma\lr{m\Omega-\tfrac1a {\hat X}_{|m|}(\pi\amn)}
    + {\hat O}^{\infty}(a^{-1})\,, 
\label{eqn:G2min-Xhat}
\end{align}
where, as in~\eqref{eqn:G1 max expand}, 
${\hat X}$ and ${\hat X}'$ are the truncations of 
\eqref{eqn:X-mz}
and 
\eqref{eqn:X-mz-prime}
as defined in Appendix~\ref{app: asymptotic expansion}, with the order of truncation suppressed but sufficiently high for what follows, and the error term ${\hat O}^{\infty}(a^{-1})$ can be made to vanish faster than any given inverse power of $a$ as $a\to\infty$ by taking the order of truncation sufficiently high. 

By the support of~$\sigma$, we may extend the range of the $n$-sum in \eqref{eqn:G2min-Xhat} to $\Nmin-1\leq n\leq\Nmax+1$. Then, the Euler-Maclaurin formula \eqref{eqn:Euler Mac} of order $l$ gives
\begin{subequations}
\label{eqn:J12}
    \begin{align}
   \GG^2_<(a)&~=~J_1+J_2
   + {\hat O}^{\infty}(a^{-1})\,, 
   \label{eq:G2less-split}
   \\\label{eqn:J1}
        J_1&~=~\frac{\pi}{a}\!\!\!\!\sum_{\substack{m\in\mc{C}^{2}\!\!,\\ m\leq M-1}}\!\!\!\int_{\Nmin-1}^{\Nmax+1}\dd x\,\frac12 J^2_{|m|}\lr{\tfrac Ra {\hat X}_{|m|}(\pi\amx)}{\hat X}'_{|m|}(\pi\amx)\sigma\lr{m\Omega-\tfrac1a {\hat X}_{|m|}(\pi\amx)}\,,\\
        J_2&~=~\frac{\pi}{2a(l!)}\!\!\!\!\sum_{\substack{m\in\mc{C}^{2}\!\!,\\ m\leq M-1}}\!\!\!\int_{\Nmin-1}^{\Nmax+1}\dd x\,\wt{B}_l(1-x)\tdn{}{x}{l} J^2_{|m|}\lr{\tfrac Ra {\hat X}_{|m|}(\pi\amx)} {\hat X}'_{|m|}(\pi\amx)\sigma\lr{m\Omega-\tfrac1a {\hat X}_{|m|}(\pi\amx)}\,.\label{eqn:J2}
    \end{align}
\end{subequations}

To bound $J_2$~\eqref{eqn:J2}, 
we recall from 
\eqref{eqn:alpha-mn} that since $x\in(\Nmin-1,\Nmax+1)$, we have to leading order ${\hat X}_{|m|}(\pi\amx)\sim\pi\amx$ and ${\hat X}_{|m|}'(\pi\amx)\sim1$. Changing variables by $s=\frac\pi a\amx$ then gives 
\begin{equation}
    J_2~\sim~\frac{1}{2(l!)}\lr{\frac\pi a}^l\sum_{\substack{m\in\mc{C}^{2}\!\!,\\ m\leq M-1}}\int_{\frac\pi a \alpha_{|m|(\Nmin-1)}}^{\frac\pi a \alpha_{|m|(\Nmax-1)}}\dd s\,\wt{B}_l\lr{1 - \tfrac a\pi s + \amo}\tdn{}{s}{l} J^2_{|m|}\lr{Rs}\sigma(m\Omega-s)\,.
\label{eq:J2-as-leadingraw}
\end{equation}
As $a\to\infty$, using in \eqref{eqn: qmn inequal} the McMahon expansion \eqref{eqn:m-McMahon expand} with $m\to|m|$ shows that 
the upper and lower limits of integration in \eqref{eq:J2-as-leadingraw} are respectively asymptotic to $\Omega m - \sigma_I$ and $\Omega m - \sigma_S$. 
Writing $s = \Omega (m + t)$, where at large $a$ we have $-\sigma_S/\Omega \lesssim t \lesssim -\sigma_I/\Omega$, 
\eqref{eq:Jsq-uniform-as-smallz} gives at large positive $m$ the leading asymptotic behaviour 
\begin{align}
J^2_{|m|}\lr{Rs} 
~=~ 
J^2_{m} \bigl(m\lr{v + t/m} \bigr)
~&\sim~
\frac{1}{2\pi\sqrt{1-\lr{v + t/m}^2}}
\frac{\ee^{-2 m\xi \lr{v + t/m}}}{m}
\notag\\
~&\sim~
\frac{1}{2\pi\sqrt{1-v^2}}
\frac{\ee^{-2 m\xi(v) + 2 t \sqrt{1-v^2}/v}}{m}
\,. 
\end{align} 
For sufficiently large~$a$, the integral in \eqref{eq:J2-as-leadingraw} is hence bounded in absolute value by a constant that is independent of $a$ and~$m$. As the number of terms in the sum over $m$ is asymptotic to ${(a\Omega)}^2$, we therefore have $J_2 = O(a^{2-l})$ as $a\to\infty$. 

To estimate $J_1$~\eqref{eqn:J1}, changing variables by $z=\frac1a {\hat X}_{|m|}(\pi\amx)$ gives 
\begin{align}
J_1 ~&=~ \sum_{\substack{m\in\mc{C}^{2}\!\!,\\ m\leq M-1}}\int_0^\infty\dd z\,\frac12J_{|m|}^2(Rz)\sigma(m\Omega-z)
\notag\\
~&=~
\sum_{m\in\mc{C}^{2}}\int_0^\infty\dd z\,\frac12J_{|m|}^2(Rz)\sigma(m\Omega-z) + O\!\left(M^{-1}\ee^{- 2 M \xi(v)}\right)
\,,
\label{eq:J1-finalestimate}
\end{align} 
where in the first equality we have extended the range of integration to the positive real line by the support of~$\sigma$, and in the second equality extended the sum to the full set~$\mc{C}^2$, noting that in the added terms the support of $\sigma$ restricts the integral to be over $m\Omega - \sigma_S \le z \le m\Omega - \sigma_I$ where $m\ge M = \lfloor(a\Omega)^2\rfloor$, whereby the error term follows using \eqref{eq:Jsq-uniform-as-smallz} with $v = R\Omega$. 

Choosing now the truncation in $\hat X$ so high that ${\hat O}^{\infty}(a^{-1})$ in 
\eqref{eqn:G2min-Xhat} is $o(a^{-3})$, and choosing $l>5$ so that $J_2 = o(a^{-3})$, combining \eqref{eq:G2less-split} and \eqref{eq:J1-finalestimate} gives 
\begin{equation}
\GG^2_<(a)~=~\sum_{m\in\mc{C}^{2}}\int_0^\infty\dd z\,\frac12J_{|m|}^2(Rz)\sigma(m\Omega-z)+o(a^{-3}).
\label{eq:G2less-final}
\end{equation}

Finally, combining 
\eqref{eq:G2 est}
and 
\eqref{eq:G2less-final}
gives 
\begin{equation}\label{eqn:G2 expand}
\GG^2(a)~=~\sum_{m\in\mc{C}^{2}}\int_0^\infty\dd z\,\frac12J_{|m|}^2(Rz)\sigma(m\Omega-z)+o(a^{-3}).
\end{equation}

\subsection{Combining $\GG^1$ and $\GG^2$}
The full expression for the leading and subleading terms for $\GG(a)$ \eqref{eqn:G vac a} is now obtained by combining $\GG^1$~\eqref{eqn:G1 expand} and $\GG^2$~\eqref{eqn:G2 expand}, with the result
\begin{align}
    \GG(a)~&=~\sum_{m\in\ZZ}\int_0^\infty\dd z\,\frac12 J_{|m|}^2(Rz)\sigma(m\Omega-z)\notag\\
&\hspace{5ex}
-\frac{R^2}{4a^3}\sum_{m\in\mc{C}^1}\lr{\sum_{n=1}^{\infty}\lrb{\frac{\pi^3}{2}\lr{n+\tfrac14}^2-\frac{3\pi}{16}-\frac{q_{1n}^2}{q_{1n}J_2^2(q_{1n})}}+\frac{\pi(5\pi^2-18)}{128}}\sigma(m\Omega)\delta_{|m|1}+\,o(a^{-3})
\notag\\
~&=~\int_{\RR}\dd\Ebar\,\sigma(\Ebar)\bigg(\frac12\sum_{m>\Ebar/\Omega}J^2_{|m|}(mv-\Ebar R)
\notag\\
&\hspace{5ex}
-\frac{R^2}{4a^3}\lr{\sum_{n=1}^\infty\lrb{\frac{\pi^3}{2}(n+\tfrac14)^2-\frac{3\pi}{16}-\frac{q_{1n}^2}{q_{1n}J^2_{2}(q_{1n})}}+\frac{\pi(5\pi^2-18)}{128}}\delta(|\Ebar|-\Omega)+o(a^{-3})
    \bigg)
\,,\label{eqn:G expand}
\end{align}
taking in the second equality everything under the integral over~$\Ebar$. This is equation \eqref{eq:G-vacuum-result} in the main text.

\section{Large-$a$ asymptotics: thermal contribution}
\label{app:large a therm}

In this appendix, we find the asymptotic behaviour of the finite temperature contribution $\Delta\GG_\beta$ \eqref{eqn:G beta a} to the integrated response function in the large-$a$ regime, leading to the results of Section~\ref{sec: large a beta sec}. We follow closely the structure of Appendix~\ref{app: large a vac}.

\subsection{Decomposition of the finite temperature contribution to the integrated response}

We decompose the finite temperature contribution $\Delta\GG_\beta$ \eqref{eqn:G beta a} to the integrated response as 
\begin{subequations}
    \begin{align}
        \Delta\GG_\beta(a)&~=~\GG^+(a)+\GG^-(a)\,,
        \label{eqn:DeltaGbeta-split}\\\label{eqn:Gplus}
        \GG^+(a)&~=~\frac{1}{a}\sum_{m,n} n(\tfrac{\beta}{a}\qmn)\frac{J^2_{|m|}\lr{\frac{R}{a}\qmn}}{\qmn J^2_{|m|+1}(\qmn)}\sigma(m\Omega-\qmn/a)\,,\\\label{eqn:Gminus}
        \GG^-(a)&~=~\frac{1}{a}\sum_{m,n} n(\tfrac{\beta}{a}\qmn)\frac{J^2_{|m|}\lr{\frac{R}{a}\qmn}}{\qmn J^2_{|m|+1}(\qmn)}\sigma(\qmn/a - m\Omega)\,,
    \end{align}
\end{subequations}
where the sums are over $m\in\ZZ$ and $n\in\NN$. 

As in Appendix~\ref{app: large a vac}, we recall that $\sigma\in C_0^\infty(\RR)$, we set $\sigma_I\coloneq \inf\supp(\sigma)$ and $\sigma_S\coloneq\sup\supp(\sigma)$, and we assume that either $\supp(\sigma)\subset\RR_{>0}$ with $0<\sigma_I<\Omega<\sigma_S$ or $\supp(\sigma)\subset\RR_{<0}$ with $\sigma_I<-\Omega<\sigma_S<0$. Since $\Delta\FF_\beta$ \eqref{eq:F beta vac KG} is even in~$\Ebar$, we shall for now assume that $\supp(\sigma)\subset\RR_{>0}$ with $0<\sigma_I<\Omega<\sigma_S$, and relax this assumption later in~\eqref{eqn: final DBG}.

For fixed $m$ in \eqref{eqn:Gplus} and~\eqref{eqn:Gminus}, the values of $n$ in the sums over $n$ are restricted by 
\begin{subequations}\label{eqn:qmn condits2}
\begin{align}
&m-\frac{\sigma_S}{\Omega}
~<~\frac{\qmn}{a\Omega}~<~m -\frac{\sigma_I}{\Omega}
\hspace{4ex}\text{in~}\GG^+\,,
\label{eqn:qmn condits2+}
\\[1ex]
&m+\frac{\sigma_I}{\Omega}
~<~\frac{\qmn}{a\Omega}~<~m +\frac{\sigma_S}{\Omega}
\hspace{4ex}\text{in~}\GG^-\,.
\label{eqn:qmn condits2-}
\end{align}
\end{subequations}
In $\GG^+$, the condition \eqref{eqn:qmn condits2+} is the same as~\eqref{eqn: qmn inequal}, and the values of $n$ for sufficiently large $a$ are hence are as listed below~\eqref{eqn: qmn inequal}, in terms of the sets $\mc{C}^{1}$ and $\mc{C}^{2}$ and the quantities $\Nmin$ and $\Nmax$ introduced therein. We shall here denote $\mc{C}^{1}$, $\mc{C}^{2}$, $\Nmin$ and $\Nmax$ with an additional superscript, by $\mc{C}^{1+}$,  $\mc{C}^{2+}$, $\Nmin^+$ and~$\Nmax^+$. In $\GG^-$, a similar analysis of \eqref{eqn:qmn condits2-} holds with $\sigma_I \to - \sigma_S$ and $\sigma_S \to - \sigma_I$, giving $\mc{C}^{1-}$,  $\mc{C}^{2-}$, $\Nmin^-$ and~$\Nmax^-$, as follows: 
\begin{itemize}
        \item[-] For $m \leq -\frac{\sigma_S}{\Omega}$, no $n$ satisfy~\eqref{eqn:qmn condits2-}.
        \item[-] For $-\frac{\sigma_S}{\Omega}<m\leq-\frac{\sigma_I}{\Omega}$, $n$ satisfies $1\leq n\leq\Nmax^-$. 
        We denote the set of these $m$ by $\mc{C}^{1-}$. 
        \item[-] For $-\frac{\sigma_I}{\Omega} < m$, $n$ satisfies $\Nmin^-\leq n\leq\Nmax^-$. 
        We denote the set of these $m$ by $\mc{C}^{2-}$. 
    \end{itemize}
We recall that the notation suppresses the dependence of $\Nmin^{\pm}$ and $\Nmax^{\pm}$ on $a$ and~$m$. We also recall that the sets $\mc{C}^{1\pm}$ are finite and note that $\pm1\in\mc{C}^{1\pm}$. 

With this notation, we may split $\GG^\pm$ as
\begin{subequations}
\label{eq:Gpm-split}
    \begin{align}
        \GG^\pm(a)&~=~\GG^{1\pm}(a)+\GG^{2\pm}(a)\,,
        \label{eq:Gpm-split-firstline}\\\label{eqn:G1pm b}
        \GG^{1\pm}(a)&~=~\frac{1}{a}\sum_{m\in\mc{C}^{1\pm}}\sum_{n=1}^{\Nmax^\pm} n(\tfrac{\beta}{a}\qmn)\frac{J^2_{|m|}\lr{\frac{R}{a}\qmn}}{\qmn J^2_{|m|+1}(\qmn)}\sigma\lr{\pm(m\Omega-\qmn/a)}\,,\\
         \GG^{2\pm}(a)&~=~\frac{1}{a}\sum_{m\in\mc{C}^{2\pm}}\sum_{n=\Nmin^\pm}^{\Nmax^\pm} n(\tfrac{\beta}{a}\qmn)\frac{J^2_{|m|}\lr{\frac{R}{a}\qmn}}{\qmn J^2_{|m|+1}(\qmn)}\sigma\lr{\pm(m\Omega-\qmn/a)}\,.\label{eqn:G2pm b}
    \end{align}
\end{subequations}
We now proceed to adapt the large $a$ analysis of $\GG(a)$ \eqref{eq:G-split} in Appendix \ref{app: large a vac} to $\GG^\pm(a)$~\eqref{eq:Gpm-split}. 
We shall see that quantitative differences arise from the small argument behaviour of the Planckian factor $n(\tfrac{\beta}{a}\qmn)$ in $\GG^\pm(a)$~\eqref{eq:Gpm-split}.

\subsection{$\GG^{1\pm}$}

We proceed with $\GG^{1\pm}$ \eqref{eqn:G1pm b} as in Section~\ref{appsec:GG1}, setting again $N\coloneq\lfloor\lr{\frac a R}^p\rfloor$ with $p\in(0,\frac14)$, and writing
\begin{subequations}
\label{eqn:G1pm def}
    \begin{align}
        \GG^{1\pm}(a)&~=~\GG^{1\pm}_<(a)+\GG^{1\pm}_>(a)\,,
        \label{eqn:G1pm-decomp}\\
        \label{eqn:G1pm smol}
        \GG^{1\pm}_<(a)&~=~\frac{1}{a}\sum_{m\in\mc{C}^{1\pm}}\sum_{n=1}^{N-1} n(\tfrac{\beta}{a}\qmn)\frac{J^2_{|m|}\lr{\frac{R}{a}\qmn}}{\qmn J^2_{|m|+1}(\qmn)}\sigma\lr{\pm(m\Omega-\qmn/a)}\,,\\
        \GG^{1\pm}_>(a)&~=~\frac{1}{a}\sum_{m\in\mc{C}^{1\pm}}\sum_{n=N}^{\Nmax^\pm} n(\tfrac{\beta}{a}\qmn)\frac{J^2_{|m|}\lr{\frac{R}{a}\qmn}}{\qmn J^2_{|m|+1}(\qmn)}\sigma\lr{\pm(m\Omega-\qmn/a)}\,,\label{eqn:G1pm lrg}
    \end{align}
\end{subequations}
using that $N<\Nmax^\pm$ for sufficiently large~$a$. 

With $\GG^{1\pm}_<$~\eqref{eqn:G1pm smol}, we proceed as with $\GG_<^1$~\eqref{eqn:g1 less}. We find that the counterpart of \eqref{eqn:G1 minus-improved} is 
\begin{equation}\label{eqn:G1pm mini}
    \GG_<^{1\pm}(a)~=~\frac{R^2}{4\beta a^2}\sigma(\Omega)\sum_{n=1}^{N-1}\frac{q_{1n}}{q_{1n}J^2_2(q_{1n})} + o(a^{-2})\,, 
\end{equation}
using that $\pm1\in\mc{C}^{1\pm}$. 

With $\GG^{1\pm}_{>}$~\eqref{eqn:G1pm lrg}, we proceed as with $\GG_>^1$~\eqref{eqn:g1 greater}. The counterpart of \eqref{eqn:I-collected} is 
\begin{subequations}
\label{eq:Ipm1234-and-defs}
\begin{equation}
    \GG_>^{1\pm}(a)~=~I_1^\pm+I_2^\pm+I_3^\pm+I_4^\pm + {\hat O}^{\infty}(a^{-1})\,,\hspace{35em}
\label{eq:Ipm1234}
\end{equation}
    \begin{align}
        I_1^\pm&=\frac{\pi}{2a}\sum_{m\in\mc{C}^{1\pm}}\int_N^\infty\dd x\,n\lr{\tfrac\beta a {\hat X}_{|m|}(\pi\amx)}J^2_{|m|}\lr{\tfrac Ra {\hat X}_{|m|}(\pi\amx)} {\hat X}'_{|m|}(\pi\amx)\sigma\lr{\pm\lr{m\Omega-\tfrac1a {\hat X}_{|m|}(\pi\amx)}}\,,\\
        I_2^\pm&=\frac{\pi}{4a}\sum_{m\in\mc{C}^{1\pm}}n\lr{\tfrac\beta a {\hat X}_{|m|}(\pi\amN)}J^2_{|m|}\lr{\tfrac Ra {\hat X}_{|m|}(\pi\amN)} {\hat X}'_{|m|}(\pi\amN)\sigma\lr{\pm\lr{m\Omega-\tfrac1a {\hat X}_{|m|}(\pi\amN)}}\,,\\
        I_3^\pm&=-\frac{\pi}{24 a}\!\!\sum_{m\in\mc{C}^{1\pm}}\!\!\!\td{}{x}\!\bigg[n\lr{\tfrac\beta a {\hat X}_{|m|}(\pi\amx)}\!J^2_{|m|}\lr{\tfrac Ra {\hat X}_{|m|}(\pi\amx)}\lrl{{\hat X}'_{|m|}(\pi\amx)\sigma\lr{\pm\lr{m\Omega-\tfrac1a {\hat X}_{|m|}(\pi\amx)}}\!\!\bigg]}{x=N}\,,\\
         I_4^\pm&=\frac{\pi}{12 a}\sum_{m\in\mc{C}^{1\pm}}\int_N^\infty\dd x\,\wt{B}_3(1-x)\tdn{}{x}{3}\bigg[n\lr{\tfrac\beta a {\hat X}_{|m|}(\pi\amx)}\nonumber\\
         &\hphantom{abdefghijklmnopqrstuvwxyzab}\times J^2_{|m|}\lr{\tfrac Ra {\hat X}_{|m|}(\pi\amx)} {\hat X}'_{|m|}(\pi\amx)\sigma\lr{\pm\lr{m\Omega-\tfrac1a {\hat X}_{|m|}(\pi\amx)}}\bigg]\,.
    \end{align}
\end{subequations}
Estimates similar to those in Section \ref{appsec:GG1} give 
\begin{subequations}\label{eqn:G1pm maxi}
    \begin{align}
        I_1^\pm&~=~I_0^\pm+\Delta I_1^\pm\,,\\
        I_0^\pm&~=~\sum_{m\in\mc{C}^{1\pm}}\int_0^\infty\dd z\,\frac12n(\beta z)J^2_{|m|}(Rz)\sigma\lr{\pm\lr{m\Omega-z}}\,,
        \label{eq:I0pm-def}\\
        \Delta I_1^\pm&~=~-\frac{R^2}{16\beta a^2}\sigma(\Omega)\lr{\pi^2(N+\tfrac14)^2-\tfrac{3}{4}}+o(a^{-2})\,,\\
        I_2^\pm&~=~\frac{\pi^2 R^2}{16\beta a^2}\sigma(\Omega)(N+\tfrac14)+o(a^{-2})\,,\\
        I_3^\pm&~=~-\frac{\pi^2 R^2}{96\beta a^2}\sigma(\Omega)+o(a^{-2})\,,\\
        I_4^\pm&~=~o(a^{-2})\,.
    \end{align}
\end{subequations}
Note that the integrals in $I_0^\pm$ \eqref{eq:I0pm-def} are well defined despite the singularity of $n(\beta z)$ at $z=0$ because $0 \notin \mc{C}^{1\pm}$. 

We choose the truncation in ${\hat X}$ so high that the error term ${\hat O}^{\infty}(a^{-1})$ in \eqref{eq:Ipm1234} is $o(a^{-2})$. From \eqref{eqn:G1pm-decomp}, 
\eqref{eq:Ipm1234}
and~\eqref{eqn:G1pm maxi}, we then find 
\begin{multline}\label{eqn:G1pm still a}
    \GG^{1\pm}(a)~=~\sum_{m\in\mc{C}^{1\pm}}\int_0^\infty\dd z\,\frac12 n(\beta z)J^2_{|m|}(Rz)\sigma\lr{\pm\lr{m\Omega-z}}\\
    -\frac{R^2}{4\beta a^2}\sigma(\Omega)\lrb{\sum_{n=1}^{N-1}\lr{\frac{\pi^2}{2}(n+\tfrac14)-\frac{q_{1n}}{q_{1n}J^2_2(q_{1n})}}+\frac{(23\pi^2-36)}{192}}+o(a^{-2})\,.
\end{multline}
To remove the $N$-dependence in~\eqref{eqn:G1pm still a}, we rewrite the sum over $n$ as
\begin{equation}
    \sum_{n=1}^{N-1}\lr{\frac{\pi^2}{2}(n+\tfrac14)-\frac{q_{1n}}{q_{1n}J^2_2(q_{1n})}}~=~\sum_{n=1}^{\infty}\lr{\frac{\pi^2}{2}(n+\tfrac14)-\frac{q_{1n}}{q_{1n}J^2_2(q_{1n})}}
    -\sum_{n=N}^{\infty}\lr{\frac{\pi^2}{2}(n+\tfrac14)-\frac{q_{1n}}{q_{1n}J^2_2(q_{1n})}}\,,\label{eqn:N depende G1pm}
\end{equation}
and note that the second sum in \eqref{eqn:N depende G1pm} is $O(N^{-1}) = O(a^{-p})$ because the summand is $O(n^{-2})$ by the McMahon expansion~\eqref{eqn:m-McMahon expand}. Hence 
\begin{multline}\label{eqn:G1pm expand}
    \GG^{1\pm}(a)~=~\sum_{m\in\mc{C}^{1\pm}}\int_0^\infty\dd z\,\frac12 n(\beta z)J^2_{|m|}(Rz)\sigma\lr{\pm\lr{m\Omega-z}}\\
    -\frac{R^2}{4\beta a^2}\sigma(\Omega)\lrb{\sum_{n=1}^{\infty}\lr{\frac{\pi^2}{2}(n+\tfrac14)-\frac{q_{1n}}{q_{1n}J^2_2(q_{1n})}}+\frac{(23\pi^2-36)}{192}}+o(a^{-2})\,, 
\end{multline}
which includes both the leading term and the first subleading term. 

\subsection{$\GG^{2\pm}$}
We proceed with $\GG^{2\pm}$ \eqref{eqn:G2pm b} as in Section~\ref{sec:G a}, setting again $M\coloneq\lfloor(a\Omega)^2\rfloor$, and assuming $a$ to be so large that $M\ge1$ and $M > 1 + \sigma_S/\Omega$. The counterpart of 
\eqref{eqn:G2decomp} reads 
\begin{subequations}
    \begin{align}
        \GG^{2\pm}(a)&~=~\GG^{2\pm}_<(a)+\GG^{2\pm}_>(a)\,,\\\label{eqn:G2min therm}
        \GG^{2\pm}_<(a)&~=~\frac{1}{a}\sum_{\substack{m\in\mc{C}^{2\pm}\!\!,\\ m\leq M-1}}\sum_{n=\Nmin^\pm}^{\Nmax^\pm}
        n(\tfrac{\beta}{a}\qmn)
        \frac{J^2_{|m|}\lr{\frac Ra\qmn}}{\qmn J^2_{|m|+1}(\qmn)}\sigma\lr{\pm\lr{m\Omega-\qmn/a}}\,,\\
        \GG^{2\pm}_>(a)&~=~\frac{1}{a}\sum_{m=M}^\infty\sum_{n=\Nmin^\pm}^{\Nmax^\pm}
        n(\tfrac{\beta}{a} q_{mn})
        \frac{J^2_{m}\lr{\frac Ra q_{mn}}}{q_{mn} J^2_{m+1}(q_{mn})}\sigma\lr{\pm\lr{m\Omega-q_{mn}/a}}\,,\label{eqn:G2max therm}
    \end{align}
\end{subequations}
where in \eqref{eqn:G2max therm} we have dropped the absolute values on $m$ because $m\ge1$. 

For $\GG^{2\pm}_>$~\eqref{eqn:G2max therm}, 
the counterpart of \eqref{eq:G2 est} has in the summand the additional factor $n(\tfrac{\beta}{a}\qmn)$, and it suffices to note here that this factor is less than unity for sufficiently large~$a$. This gives $\GG^{2\pm}_>(a) = O\lr{\ee^{-2M\xi(v)}} = o(a^{-2})$. 

For $\GG^{2\pm}_<$~\eqref{eqn:G2min therm},
the counterpart of \eqref{eqn:J12} reads 
\begin{subequations}
    \begin{align}
        \GG^{2\pm}_<(a)&~=~J^\pm_1+J^\pm_2 + {\hat O}^{\infty}(a^{-1})
        \,,
        \label{eqn:J12pm-decomp}
        \\\label{eqn:J1 therm}
        J^\pm_1&~=~\frac{\pi}{a}\sum_{\substack{m\in\mc{C}^{2\pm}\!\!,\\ m\leq M-1}}\int_{\Nmin^\pm-1}^{\Nmax^\pm+1}\dd x\,\frac12 n\lr{\tfrac\beta a {\hat X}_{|m|}(\pi\amx)}J^2_{|m|}\lr{\tfrac Ra {\hat X}_{|m|}(\pi\amx)}\nonumber\\
        &\hphantom{abcdefghijklmnopqrstuvwxyzabcdefg}\times {\hat X}'_{|m|}(\pi\amx)\sigma\lr{\pm\lr{m\Omega-\tfrac1a {\hat X}_{|m|}(\pi\amx)}}\,,\\
        J^\pm_2&~=~\frac{\pi}{2a(l!)}\sum_{\substack{m\in\mc{C}^{2\pm}\!\!,\\ m\leq M-1}}\int_{\Nmin^\pm-1}^{\Nmax^\pm+1}\dd x\,\wt{B}_l(1-x)\tdn{}{x}{l}\bigg[n\lr{\tfrac\beta a {\hat X}_{|m|}(\pi\amx)}\nonumber\\
        &\hphantom{abcdefghijklmnopqr}\times J^2_{|m|}\lr{\tfrac Ra {\hat X}_{|m|}(\pi\amx)} {\hat X}'_{|m|}(\pi\amx)\sigma\lr{\pm\lr{m\Omega-\tfrac1a {\hat X}_{|m|}(\pi\amx)}}\bigg]\,.\label{eqn:J2 therm}
    \end{align}
\end{subequations}
To bound $J^\pm_2$~\eqref{eqn:J2 therm}, the counterpart of \eqref{eq:J2-as-leadingraw} reads 
\begin{equation}
    J^\pm_2~\sim~\frac{1}{2(l!)}\lr{\frac{\pi}{a}}^l\sum_{\substack{m\in\mc{C}^{2\pm}\!\!,\\ m\leq M-1}}\int_{\frac\pi a \alpha_{|m|(\Nmin^\pm-1)}}^{\frac\pi a \alpha_{|m|(\Nmax^\pm+1)}}\dd s\,\wt{B}_l\lr{1 - \tfrac a\pi s + \amo}\tdn{}{s}{l}\bigg[n(\beta s)J^2_{|m|}(Rs)\sigma\lr{\pm\lr{m\Omega-s}}\bigg]\,.
\label{eq:J2pm-as-leadingraw}
\end{equation}
Arguments similar to those below \eqref{eq:J2-as-leadingraw} then show that $J^\pm_2 = O(a^{2-l})$ as $a\to\infty$. 
For $J^\pm_1$~\eqref{eqn:J1 therm}, 
arguments similar to those in \eqref{eq:J1-finalestimate} give 
\begin{equation}
J^\pm_1~=~\sum_{m\in\mc{C}^{2\pm}}\int_0^\infty\dd z\,\frac12 n(\beta z)J^2_{|m|}(Rz)\sigma\lr{\pm\lr{m\Omega-z}}
    + O\!\left(M^{-1}\ee^{- 2 M \xi(v)}\right)\,.
\label{eq:J1pm-finalestimate}
\end{equation}
The integrals in \eqref{eq:J2pm-as-leadingraw} and \eqref{eq:J1pm-finalestimate} are well defined despite the singularity of $n(x)$ at $x=0$, for $J^+_1$ and $J^+_2$ because $0\notin \mc{C}^{2+}$, and for $J^-_1$ and $J^-_2$ because $0\in \mc{C}^{2-}$ but 
$0 \notin \supp(\sigma)$. 

Choosing the truncation in $\hat X$ so high that ${\hat O}^{\infty}(a^{-1})$ in 
\eqref{eqn:J12pm-decomp} $o(a^{-2})$, and choosing $l>4$ so that $J_2 = o(a^{-2})$, combining \eqref{eqn:J12pm-decomp} and \eqref{eq:J1pm-finalestimate} gives 
\begin{equation}\label{eqn:G2pm expand}
    \GG^{2\pm}(a)~=~\sum_{m\in\mc{C}^{2\pm}}\int_0^\infty\dd z\,\frac12 n(\beta z)J^2_{|m|}(Rz)\sigma\lr{\pm\lr{m\Omega-z}}+o(a^{-2})\,.
\end{equation}

\subsection{Combining $\GG^{1\pm}$ and $\GG^{2\pm}$}

Combining $\GG^{1\pm}$ \eqref{eqn:G1pm expand} and $\GG^{2\pm}$ in~\eqref{eq:Gpm-split-firstline}, \eqref{eqn:DeltaGbeta-split}
gives 
\begin{align}
    \Delta\GG_\beta(a)~&=~\sum_{m\in\ZZ}\frac12\int_0^\infty\dd z\,n(\beta z)J_{|m|}^2(Rz)\lrb{\sigma(m\Omega-z)+\sigma(-(m\Omega-z))}
    \notag\\
    & \hspace{4ex}-\frac{R^2}{2\beta a^2}\sigma(\Omega)\lrb{\sum_{n=1}^\infty\lr{\frac{\pi^2}{2}(n+\tfrac14)-\frac{q_{1n}}{q_{1n}J_2^2(q_{1n})}}+\frac{(23\pi^2-36)}{192}}+o(a^{-2})
    \notag\\
~&=~ \int_{\RR}\dd\Ebar\,\sigma(\Ebar)
    \bigg(\frac{1}{2}\sum_{m>\Ebar/\Omega}n(\beta\omega_+)J^2_{|m|}(\omega_+R)
    +\frac{1}{2}\sum_{m>-\Ebar/\Omega}n(\beta\omega_-)J^2_{|m|}(\omega_-R)\notag\\
    & \hspace{7ex}
    -\frac{R^2}{2\beta a^2}\lrb{\sum_{n=1}^\infty\lr{\frac{\pi^2}{2}(n+\tfrac14)-\frac{q_{1n}}{q_{1n}J_2^2(q_{1n})}}+\frac{(23\pi^2-36)}{192}}\delta(\Ebar-\Omega) +o(a^{-2})
    \bigg)\,,
    \label{eqn:almost final DBG}
\end{align}
taking in the second equality everything under the integral over $\Ebar$, and writing 
$\omega_{\pm}=m\Omega\mp\Ebar$.

Finally, we recall that the result \eqref{eqn:almost final DBG} was obtained under the assumption $\supp(\sigma)\subset\RR_{>0}$ with $0<\sigma_I<\Omega<\sigma_S$. 
Including now also the case where $\supp(\sigma)\subset\RR_{<0}$ with $\sigma_I<-\Omega<\sigma_S<0$, 
\eqref{eqn:almost final DBG} is replaced by
\begin{multline}\label{eqn: final DBG}
    \Delta\GG_\beta(a)~=~\int_{\RR}\dd\Ebar\,\sigma(\Ebar)
    \bigg(\frac{1}{2}\sum_{m>|\Ebar|/\Omega}n(\beta\omega_+)J^2_{|m|}(\omega_+R)
    +\frac{1}{2}\sum_{m>-|\Ebar|/\Omega}n(\beta\omega_-)J^2_{|m|}(\omega_-R)\\
    -\frac{R^2}{2\beta a^2}\lrb{\sum_{n=1}^\infty\lr{\frac{\pi^2}{2}(n+\tfrac14)-\frac{q_{1n}}{q_{1n}J_2^2(q_{1n})}}+\frac{(23\pi^2-36)}{192}}\delta(|\Ebar|-\Omega) +o(a^{-2})
    \bigg)\,. 
\end{multline}
This is equation \eqref{eq:DeltaGbeta-integral-maintext} in the main text.

\twocolumngrid
\bibliography{zz_bibliography}

\end{document}